\documentclass[12pt,a4paper]{article}
\usepackage{enumitem}
\usepackage[margin=2.2cm]{geometry}
\usepackage{authblk}

\usepackage[breaklinks]{hyperref}
\usepackage{booktabs}
\usepackage{newtxtext,newtxmath}
\usepackage{soul}
\usepackage[colorinlistoftodos,textsize=footnotesize]{todonotes}

\usepackage{amsmath,amssymb,amsthm,mathbbol,mathrsfs}
   \usepackage{dsfont}
\let\Re\relax
\let\Im\relax
\DeclareMathOperator{\Re}{Re}
\DeclareMathOperator{\Im}{Im}

\newcommand{\1}{\mathds{1}}                         
\newcommand{\NN}{\mathbb{N}}          
\newcommand{\ZZ}{\mathbb{Z}} 
\newcommand{\RR}{\mathbb{R}}           
\newcommand{\CC}{\mathbb{C}}           
\newcommand{\M}{\mathbb{M}} 	     

\newcommand{\Hcal}{\mathcal {H}}
\newcommand{\Lcal}{\mathcal {L}}
\newcommand{\Bcal}{\mathcal {B}}

\newcommand{\F}{\mathcal{F}}

\newcommand{\Ocal}{\mathcal{O}}

\newcommand{\fF}{\mathfrak{F}}

\newcommand{\MM}{\mathbb{M}}
\newcommand{\C}{\mathbb{C}}
\newcommand{\BB}{{\mathcal{B}}}
\newcommand{\FF}{{\mathfrak{F}}}
\newcommand{\HH}{{\mathcal{H}}}

\newcommand{\DD}{{\mathscr{D}}}

\newcommand{\Mf}{{\mathfrak{M}}}
\newcommand{\Nf}{{\mathfrak{N}}}

\newcommand{\II}{{\mathbb{1}}}

\newcommand{\kb}{{\boldsymbol{k}}}
\newcommand{\xb}{{\boldsymbol{x}}}
\newcommand{\yb}{{\boldsymbol{y}}}

\newcommand{\Ac}{{\mathcal{A}}}
\newcommand{\Cc}{{\mathcal{C}}}
\newcommand{\Bc}{{\mathcal{B}}}
\newcommand{\Fc}{{\mathcal{F}}}
\newcommand{\Ic}{{\mathcal{I}}}
\newcommand{\Jc}{{\mathcal{J}}}
\newcommand{\Pc}{{\mathcal{P}}}
\newcommand{\Uc}{{\mathcal{U}}}
\newcommand{\Vc}{{\mathcal{V}}}

\newcommand{\Wc}{{\mathcal{W}}}

\newcommand{\Ws}{{\mathsf{W}}}

\newcommand{\id}{{\textrm{id}}}
\newcommand{\ph}{\varphi}

\newcommand{\dvol}{{\textrm{dvol}}}

\newcommand{\CoinX}[1]{C_0^\infty(#1)}

\newcommand{\ip}[2]{\langle #1|#2\rangle}
\newcommand{\ket}[1]{|#1\rangle}
\newcommand{\bra}[1]{\langle #1|}

\newcommand{\be}{\begin{equation}}
\newcommand{\ee}{\end{equation}}

\newtheorem{theorem}{Theorem}
\newtheorem{proposition}[theorem]{Proposition}
\newtheorem{lemma}[theorem]{Lemma}
\newtheorem{corollary}[theorem]{Corollary}

\newtheorem*{remark*}{Remark}
\newtheorem{example}[theorem]{Example}
\newtheorem{exercise}[theorem]{Exercise}

{\theoremstyle{definition}
\newtheorem{definition}[theorem]{Definition}
\newtheorem{df}[theorem]{Definition}}

\DeclareMathOperator{\im}{im}
\DeclareMathOperator{\tr}{tr}

\DeclareMathOperator{\Prob}{Prob}
\DeclareMathOperator{\supp}{supp}

\DeclareMathOperator{\Sol}{Sol}
\begin{document}
	\title{Algebraic Quantum Field Theory -- an introduction} 

	\author[1]{Christopher J Fewster\thanks{\tt chris.fewster@york.ac.uk}}
	\author[1]{Kasia Rejzner\thanks{\tt kasia.rejzner@york.ac.uk}}
	\affil{Department of Mathematics,
		University of York, Heslington, York YO10 5DD, United Kingdom.}
 
	\date{\today}
	\maketitle 
	%
	%
	\begin{abstract}  
	We give a pedagogical introduction to algebraic quantum field theory (AQFT), with the aim of explaining its key structures and features. Topics covered include: algebraic formulations of quantum theory and the GNS representation theorem, the appearance of unitarily inequivalent representations in QFT (exemplified by the van Hove model), the main assumptions of AQFT and simple models thereof, the spectrum condition, Reeh--Schlieder theorem, split property, the universal type of local algebras, and the theory of superselection sectors. The abstract discussion is illustrated by concrete examples. One of our concerns is to explain various ways in which quantum field theory differs from quantum mechanics, not just in terms of technical detail, but in terms of physical content. 
	The text is supplemented by exercises and appendices that enlarge on some of the relevant mathematical background. These notes are based on lectures given by CJF for the International Max Planck Research School at the Albert Einstein Institute, Golm (October, 2018) and by KR at the Raman Research Institute, Bangalore (January, 2019). 
	\end{abstract} 
	\tableofcontents 
	
	\section{Introduction}
		
	Algebraic Quantum Field Theory (AQFT) is one of two axiomatic programmes for QFT that emerged in the 1950s, in response to the problem of making QFT mathematically precise.
	While Wightman's programme~\cite{StreaterWightman} maintains an emphasis on quantum fields, AQFT~\cite{Haag,Araki}, developed initially by Haag, Kastler, Araki and others, takes the more radical step on focussing on local observables, with the idea that fields can emerge as natural ways of labelling some of the observables. Like Wightman theory, its primary focus is on setting out a precise mathematical framework into which all QFTs worthy of the name should fit. This permits one to separate the general study of the structure and properties of QFT from the problem of constructing (by whatever means) specific QFT models that obey the assumptions. 
	The early development of AQFT is well-described in the monographs of Haag~\cite{Haag} and Araki~\cite{Araki}. Mathematically, it makes extensive use of operator algebra methods and indeed has contributed to the theory of von Neumann algebras in return. Relevant aspects of operator algebra theory, with links to the physical context, can be found in the monographs of Bratteli and Robinson~\cite{BratRob:vol1,BratRob:vol2}.
	AQFT also comprises a lot of machinery for treating specific QFT models, which have some advantages relative to other approaches to QFT. 
	During the last 20 years it has also been adapted to provide rigorous constructions of perturbative QFT, and also of some low-dimensional models, and its overall viewpoint
	has been particularly useful in the theory of quantum fields in curved spacetimes. A recent edited collection~\cite{AdvAQFT} summarises these developments, and the two recent monographs~\cite{Rejzner_book,duetsch2019classical} in particular describe the application to perturbation theory, while~\cite{hollands2018entanglement} concerns entanglement measures in QFT. An extensive survey covering some of the topics presented here in much greater depth can be found in~\cite{HalvorsenMueger:2006}.
	
	The purpose of these lectures is to present an introduction to AQFT
	that emphasises some of its original motivations and de-mystifies some of its terminology (GNS representations, spectrum condition, Reeh--Schlieder, type III factors, split property, superselection sectors...). We also emphasise features of QFT that sharply distinguish it from quantum mechanics and which can be seen particularly clearly in the AQFT framework. Our treatment is necessarily limited and partial; the reader is referred to the literature mentioned for more detail and topics not covered here. 
	
	The idea of \emph{algebraic} formulations of quantum theory, which we describe in Section~\ref{sec:algQM}, can be traced
	back to Heisenberg's matrix mechanics, in which the algebraic relations between
	observables are the primary data. Schr\"odinger's wave mechanics, by contrast,
	starts with spaces of wavefunctions, on which the observables of the theory act in specific ways.
	
	As far as position and momentum go, and for systems of finitely many degrees of freedom, the distinction is rather inessential, because the Stone-von Neumann theorem guarantees that any (sufficiently regular\footnote{To deal with the technical problems of using unbounded operators.}) irreducible representation
	of the commutation relations is \emph{unitarily equivalent} to the Schr\"odinger representation. However, the angular momentum operators provide a classic example in which inequivalent physical representations appear, and it is standard
	to study angular momentum as an algebraic representation problem. 
	However, it was a surprise in the development of QFT that unitarily inequivalent representations have a role to play here, and indeed turn out to be ubiquitous. Section~\ref{sec:vH} is devoted to the van Hove model, one of the first examples in which this was understood. The van Hove model concerns a free scalar field with an external source, and is explicitly solvable. However, one can easily find situations in which a naive interaction picture approach fails to reproduce the correct solution -- a failure that can be clearly ascribed to a failure of unitary equivalence between different representations of the canonical commutation relations (CCRs).
	
	After these preliminaries, we set out the main assumptions of Algebraic Quantum Field Theory in Sec.~\ref{sec:AQFT}. In fact there are many variants of AQFT and we give a liberal set of axioms that can be strengthened in various ways. We also describe how some standard QFT models can be formulated in terms of AQFT. Although we focus on free theories, it is important to emphasise that AQFT is intended as a framework for all quantum field theories worthy of the name, and successfully encompasses some
	nontrivial interacting models in low dimensions. AQFT distinguishes between two levels of structure: on the one hand, the algebraic relations among observables and on the other, their concrete realisation as Hilbert space operators. The link is provided by the GNS representation theorem (described in Sec.~\ref{sec:GNS}) once a suitable state has been given. For this reason we spend some time on states of the free scalar field, describing in particular the \emph{quasi-free} states, which have representations on suitable Fock spaces. These include the standard vacuum state as well as thermal states.
	
	The remaining parts of the notes concern general features of AQFT models, where the power of the technical framework begins to come through. Among other things, we prove the Reeh--Schlieder theorem and discuss some of its consequences in Section~\ref{sec:spectrum}, before turning in Section~\ref{sec:universaltype} to the structure of the local von Neumann algebras in suitable representations and the remarkable result (which we describe, but do not prove) that they are all isomorphic to the unique hyperfinite factor of type III${}_1$. The distinction between one theory and another therefore lies in the way these algebras are situated, relative to one another, within the algebra of bounded operators on the Hilbert space of the theory. Finally, Sections~\ref{sec:split} and~\ref{sec:sectors} discuss the split property and the theory of superselection sectors. Like the theory in Section~\ref{sec:universaltype}, these are deep and technical subjects and our aim here is to present the main ideas and some outline arguments, referring the dedicated reader to the literature. 
	On the subject of literature: in this pedagogical introduction we have tended to give references to monographs rather than original papers, so the reference list is certainly not intended as a comprehensive survey of the field. 
	
	These notes represent a merger and expansion of lectures given by CJF at the AEI in Golm (October, 2018) and by KR at the Raman Research Institute, Bangalore (January, 2019). We are grateful to the students and organisers of the lecture series concerned. We are also grateful to the organisers of the conference \emph{Progress and Visions in Quantum Theory in View of Gravity} (Leipzig, October 2018) for the opportunity to contribute to their proceedings.
	
\section{Algebraic quantum mechanics}\label{sec:algQM}
\subsection{Postulates of quantum mechanics}\label{sec:PosQM}
The standard formalism of quantum theory starts with a complex Hilbert space $\Hcal$, whose elements $\phi\in\Hcal$ are called \textit{state vectors}. (For convenience some basic definitions concerning operators on Hilbert space are collected in Appendix~\ref{appx:basic_fa}.)
The key postulates of quantum mechanics say that:
\begin{itemize}
\item \textbf{Pure states} of a quantum system are described by \textit{rays} in $\Hcal$, i.e. $[\psi]:=\{\lambda \psi|\lambda\in\C\}$. \textbf{Mixed states} are described by density matrices, i.e., positive \textit{trace-class operators} $\rho:\Hcal\rightarrow\Hcal$, with unit trace. 
\item \textbf{Observables} are described by \textit{self-adjoint operators} $A$ on $\Hcal$. However, the self-adjoint operators corresponding to observables may be a proper subset of the self-adjoint operators on $\Hcal$; in particular, this occurs if the system is subject to superselection rules.
\end{itemize}

The probabilistic interpretation of quantum mechanics\footnote{For simplicity, we restrict to sharp measurements, avoiding the introduction of positive operator valued measures.} is based on the idea that one can associate to a self-adjoint operator $A$ and a normalised state vector $\psi\in\Hcal$ a probability measure $\mu_{\psi,A}$, so that the probability of the measurement outcome to be within a Borel subset $\Delta\subset\RR$ (for instance, an interval $[a,b]$) is given by 
\[
\Prob(A\in\Delta;\psi)=\int_\Delta\mu_{\psi,A}(\lambda)= \ip{\psi}{P_{A}(\Delta)\psi}\,,
\]
where $P_{A}(\Delta)$ is the \emph{spectral projection} of the operator $A$ associated with $\Delta$, and indeed $\Delta\to P_{A}(\Delta)$ determines a \emph{projection-valued measure}. The probability measure $\mu_{\psi,A}$ depends only on the ray $[\psi]$ and has support contained within the spectrum of $A$. 
The moments of this measure are given by
\[
\nu_n:=\int \lambda^n\mu_{\psi,A}(\lambda)= \left<\psi,A^n\psi\right>\,;
\] 
conversely, the moments determine the measure uniquely subject to certain growth conditions. For example, the Hamburger moment theorem~\cite{Simon:1998}\index{Hamburger moment theorem} guarantees uniqueness provided that there are constants $C$ and $D$ such that
\begin{equation}
|\nu_n|\le C D^n n!\qquad\text{for all $n\in\NN_0$.}
\end{equation} 
Note, however, that
there are many examples in quantum mechanics in which the moments grow too fast for the unique reconstruction of a probability measure. 

\begin{example} 
Consider the quantum particle confined to an interval $(-a,a)$ subject to either Dirichlet or Neumann boundary conditions at the endpoints, with corresponding Hamiltonian operators $H_D$ or $H_N$ respectively. Measurements of the energy in a state $\psi\in L^2(-a,a)$ supported away from the endpoints are distributed according to probability distributions $\mu_{\psi,H_D}$ and $\mu_{\psi,H_N}$, which differ because the spectra of $H_D$ and $H_N$ differ. However, they share a common moment sequence because $H_D$ and $H_N$ agree on state vectors supported away from the endpoints.
\end{example}

One can combine effects from two physical states described by state vectors $\psi_1$ and $\psi_2$ by building their \textit{superposition}, which, however depends on the choice of representative state vectors, since the ray corresponding to 
\[
\psi=\alpha \psi_1+\beta\psi_2
\] 
typically depends on the choice of $\alpha$ and $\beta$. However, sometimes the relative phase between the state vectors we are superposing cannot be observed. For example, this occurs if $\psi_1$ is a state of integer angular momentum, while $\psi_2$ is a half-integer angular momentum state. The physical reason for this is that a $2\pi$-rotation cannot be distinguished from no rotation at all. Of course there are self-adjoint operators on the Hilbert space that do sense the relative phase: the point is that these operators are not physical observables. 

Let us give a brief argument for the existence of superselection sectors when the theory possesses a charge $Q$, which is supposed to be conserved in any interactions available to measure it. For simplicity, we assume that $Q$ has discrete spectrum. Let $\psi$ be any eigenstate $Q\psi=q\psi$ and let $P$ be any projection corresponding to a zero-one measurement. After an ideal measurement of $P$ in state $\psi$ returning the value $1$, the system is in state $P\psi$. But as charge is conserved, $P\psi$ must be an eigenstate of $Q$ with eigenvalue $q$, so
\[
QP\psi = q P\psi = PQ\psi.
\]
We deduce that $[Q,P]\psi=0$ and, as the Hilbert space is spanned by eigenstates of $Q$, it follows that $Q$ and $P$ commute. Furthermore, $Q$ commutes with every self-adjoint operator representing a physical observable because any spectral projection of such an operator also corresponds to a physical observable.  

Mathematically, the allowed physical observables are all block diagonal with respect to a 
decomposition of the Hilbert space $\Hcal$ as
\[
\Hcal=\bigoplus_{i\in I} \Hcal_i\,,
\]
where $I$ is some index set and the subspaces $\Hcal_i$ are called \textit{superselection sectors}, which would be the charge eigenspaces in our example above. The relative phases between state vectors belonging to different sectors cannot be observed. One of the main motivations behind AQFT was to understand how superselection sectors arise in QFT. We will see that the different superselection sectors correspond to unitarily inequivalent representations of the algebra of observables. 
\subsection{Algebraic approach}\label{sec:AQM}
The main feature of the algebraic description of quantum theory is that the focus shifts from states to observables, and their algebraic relations. It is worth pausing briefly to consider the motivation for an algebraic description of observables -- this is a long story if told in full (see~\cite{Emch}), but one can explain the essential elements quite briefly. 

The central issue is to provide an operational meaning for the linear combination and product of observables. Let us suppose that a given observable is measured by a certain instrument; for measurements conducted in each particular state, the numerical readout on the instrument is statistically distributed in a certain way, so the observable may be thought of as a mapping from states to random variables taking values in $\RR$. Given two such observables, we can form a third, by taking a fixed linear combination of the random variables concerned, restricting to real-linear combinations in the first instance. So there is a clear justification for treating the set of observables as a real vector space. Similarly, we may apply a function to an observable by applying it to the random variables concerned; this may be regarded as repainting the scale on the measuring instruments. Given any
two observables $A$ and $B$ it is now possible to form the observable
\begin{equation}
A\circ B:= \frac{1}{2}\left((A+B)^2 - A^2 - B^2\right),
\end{equation} 
simply by forming linear combinations and squares. This may be regarded as a symmetrised product of $A$ and $B$. The remaining problem, which is naturally where the hard work lies, is to find appropriate additional conditions under which the vector space of observables can be identified with the self-adjoint elements of a $*$-algebra, so that $A\circ B = AB + BA$. We refer the reader to~\cite{Emch}; however, it is clear that observables naturally admit some algebraic structure beyond that of a vector space.

The main postulates of quantum theory, in its algebraic form, are now formulated as follows: 
\begin{enumerate}
	\item A physical system is described by a unital $*$-algebra $\Ac$, 
	whose self-adjoint elements are interpreted as the \textbf{observables}. It is conventional though slightly imprecise to call $\Ac$ the \textit{algebra of observables}.
	In many situations we impose the stricter condition that $\Ac$ be a unital $C^*$-algebra.
	\item \textbf{States} are identified with positive, normalized linear functionals  $\omega:\Ac\to\CC$, i.e. we require $\omega(A^*A)\geq 0$ for all $A\in\Ac$ and $\omega(\1)=1$ as well as $\omega$ being linear. The state is
	\emph{mixed} if it is a convex combination of distinct states (i.e., $\omega=\lambda\omega_1+(1-\lambda)\omega_2$ with $\lambda\in(0,1)$, $\omega_1\neq \omega_2$) and \emph{pure} otherwise. 
\end{enumerate}
Some definitions are in order. 
\begin{df}
	A \emph{$*$-algebra} (also called an involutive complex algebra) $\Ac$ is an algebra over $\C$, together with a map, ${}^*:\Ac \rightarrow \Ac$, called an involution, with the following properties:
	\begin{enumerate}
		\item    for all $A, B \in \Ac$: $(A + B)^* = A^* + B^*$, $(A B)^* = B^* A^*$,
		\item   for every $\lambda\in\CC$ and every $A \in \Ac$: $(\lambda A)^* = \overline{\lambda} A^*$,
		\item    for all $A \in \Ac$: $(A^*)^* = A$.
	\end{enumerate}
	The $*$-algebra is \emph{unital} if it has an element $\II$ which is a unit for the algebraic product ($A\II=\II A=A$ for all $A\in\Ac$) and is therefore invariant under the involution. Unless explicitly indicated otherwise, a homomorphism $\alpha:\Ac_1\to\Ac_2$ between two $*$-algebras will be understood to be an algebraic homomorphism that respects the involutions ($(\alpha A)^*=\alpha (A^*)$) and preserves units ($\alpha \II_{\Ac_1}=\II_{\Ac_2}$).  
\end{df}
 The bounded operators $\Bcal(\Hcal)$ on a Hilbert space $\Hcal$ form a $*$-algebra, with the adjoint as the involution, but there are other interesting examples. 
\begin{exercise}\label{ex:LDH}
(Technical -- for those familiar with unbounded operators.) Given a dense subspace $\DD$ of a Hilbert space $\HH$, let $\Lcal(\DD,\HH)$
be the set of all (possibly unbounded) operators $A$ on $\HH$ defined on, and leaving invariant, $\DD$, (i.e., $D(A)=\DD$, $A\DD\subset\DD$) and having an adjoint with $\DD\subset D(A^*)$. Then $\Lcal(\DD,\HH)$ may be identified with a subspace of the vector space of all linear maps from $\DD$ to itself. Verify that $\Lcal(\DD,\HH)$ is an algebra with respect to composition of maps and that the map $A\mapsto A^*|_\DD$ is an involution on $\Lcal(\DD,\HH)$, making it a $*$-algebra. Show also that $\Lcal(\HH,\HH)=\BB(\HH)$. (Hint: Use the Hellinger--Toeplitz theorem~\cite[\S III.5]{ReedSimon:vol1}.)
\end{exercise}  
The algebra of bounded operators also carries a norm that is compatible with the algebraic structure in various ways. In general we can make the following definitions: 
\begin{df}
	A \emph{normed algebra} $\Ac$ is an algebra equipped with a norm $\|.\|$ satisfying 
	\[
	\|AB\|\leq\|A\|\|B\|\,.
	\] 
	If $\Ac$ is unital, then it is a \emph{normed unital algebra} if in addition $\|\1\|=1$. If $\Ac$ is complete in the topology induced by $\|\cdot\|$ then 
	$\Ac$ is a \emph{Banach algebra}; if, additionally, $\Ac$ is a $*$-algebra and $\|A^*\|=\|A\|$, then $\Ac$ is a \emph{Banach $*$-algebra} or $B^*$-algebra. 
\end{df} 
A \emph{$C^*$-algebra} is a particular type of $B^*$-algebra.  
\begin{df}
A \emph{$C^*$-algebra} $\Ac$ is a $B^*$-algebra whose norm has the $C^*$-property:
	\[
	\|A^* A \| = \|A\|\|A^*\| = \|A\|^2, \quad\forall A\in \Ac\,.
	\]
\end{df}
The bounded operators $\Bcal(\Hcal)$, with the operator norm, provide an important example of a $C^*$-algebra. A useful property of unital $C^*$-algebras is that homomorphisms between them are automatically continuous~\cite[Prop.~2.3.1]{BratRob:vol1}, with unit norm. 

Turning to our second postulate,  
the role of the state in the algebraic approach is to assign expectation values:
if $A=A^*$, we interpret $\omega(A)$ as the expected value of $A$ if measured in the state $\omega$. At first sight this definition seems far removed from the notion of a state in conventional formulations of quantum mechanics. Let us see that it is in fact a natural generalisation.

Suppose for simplicity (and to reduce notation) that $\Ac$ is an algebra of bounded operators acting on a Hilbert space $\HH$, with the unit of $\Ac$ coinciding with the unit operator on $\HH$. Then every unit vector $\psi\in\HH$ induces a \emph{vector state} on $\Ac$ by the formula
\[
\omega_\psi(A) = \ip{\psi}{A\psi},
\] 
as is seen easily by computing $\ip{\psi}{\II\psi}=1$ and $\ip{\psi}{A^*A\psi}=\|A\psi\|^2\ge 0$. 
\begin{exercise}
	Show that every density matrix (a positive trace-class operator $\rho$ on $\HH$ with $\tr\rho=1$) induces a state on $\Ac$ according to 
	\[
	\omega_\rho(A) = \tr\rho A\qquad A\in\Ac.
	\]
\end{exercise}
However it is important to realise that, in general, not all algebraic states on $\Ac$ need arise from vectors or density matrices in a given Hilbert space representation. 
A further important point is that the definition of a state is purely mathematical in nature. It is not guaranteed that all states correspond to physically realisable situations, and indeed a major theme of the subject is to identify classes of states and representations that, by suitable criteria, may be regarded as physically acceptable.

\begin{exercise}
	By mimicking the standard arguments from quantum mechanics or linear algebra, show that every state $\omega$ on a $*$-algebra $\Ac$ induces a Cauchy--Schwarz inequality
	\begin{equation}\label{eq:CauchySchwarz}
	|\omega(A^*B)|^2\le \omega(A^*A)\omega(B^*B)
	\end{equation}
	for all $A,B\in\Ac$. Show also that $\omega(A^*)=\overline{\omega(A)}$, for any $A\in\Ac$. (Hint: consider the linear combination $\II+\alpha A$, for $\alpha\in\CC$.)
\end{exercise}

\subsection{States and representations}\label{sec:GNS}
The Hilbert space formulation of quantum mechanics is too useful to be abandoned entirely and the study of Hilbert space representations forms an important part of AQFT. Let us recall a few definitions. 
\begin{df}
	A \emph{representation} of a unital $*$-algebra $\Ac$ consists of
	a triple $(\HH,\DD,\pi)$, where $\HH$ is a Hilbert space, $\DD$ a dense subspace of $\HH$, and $\pi$ a map from $\Ac$ to operators on $\HH$ with the following properties: 
	\begin{itemize}
	    \item each $\pi(A)$ has domain $D(\pi(A))=\DD$ and range contained in $\DD$, 
	    \item $\pi(\II)=\II|_\DD$,
	    \item $\pi$ respects linearity and products,
	    \[
	    \pi(A+\lambda B+CD) = \pi(A) + \lambda\pi(B) + \pi(C)\pi(D), \qquad A,B,C,D\in\Ac,~\lambda\in\CC
	    \]
	    \item each $\pi(A)$ has an adjoint with $\DD\subset D(\pi(A)^*)$,  whose restriction to $\DD$ obeys $\pi(A)^*|_\DD=\pi(A^*)$.
	\end{itemize}
	In short, $\pi$ is a homomorphism  
	from $\Ac$ into the $*$-algebra $\Lcal(\DD,\HH)$ defined in  Exercise~\ref{ex:LDH}. Note that every $\pi(A)$ is closable, due to the fact that $\pi(A)^*$ is densely defined. We will also use the shorthand notation $(\HH,\pi)$ for a representation $(\HH,\HH,\pi)$. In this case, 
	$\pi$ is a homomorphism	$\pi:\Ac\rightarrow\Bcal(\Hcal)$, and is necessarily continuous if $\Ac$ is a $C^*$-algebra.
	\\
	
	\noindent A representation $\pi$ is called \emph{faithful} if $\ker \pi=\{0\}$. It is called irreducible if there are no subspaces of $\Hcal$ invariant under $\pi(\Ac)$ that are not either trivial or dense in $\Hcal$.\\
\end{df} 
\begin{df}
	Two representations $(\HH_1,\DD_1,\pi_1)$ and $(\HH_2,\DD_2,\pi_2)$ of a $*$-algebra $\Ac$ are called \emph{unitarily equivalent}, if
	there is a unitary map $U:\HH_1\rightarrow\HH_2$ which 
	restricts to an isomorphism between $\DD_1$ and $\DD_2$, and 
	$U\pi_1(A)=\pi_2(A) U$ holds for all $A\in\Ac$. They are \emph{unitarily inequivalent} if they are not unitarily equivalent.
\end{df} 

On a first encounter, algebraic states feel unfamiliar because one is so used to the Hilbert space version. However algebraic states are not too far away from a Hilbert space setting. The connection is made by the famous GNS (Gel'fand, Naimark, Segal) representation theorem. 
\begin{theorem}  Let $\omega$ be a state on a unital $*$-algebra $\Ac$. 
	Then there is a representation $(\HH_\omega,\DD_\omega,\pi_\omega)$
	of $\Ac$ and a unit
	vector $\Omega_\omega\in\DD_\omega$ such that $\DD_\omega=\pi_\omega(\Ac)\Omega_\omega$ and
	\begin{equation}\label{eq:GNSprop}
	\omega(A) =\ip{\Omega_\omega}{\pi_\omega(A)\Omega_\omega}, \qquad \forall~A\in \Ac.
	\end{equation}
	Furthermore $(\HH_\omega,\DD_\omega,\pi_\omega,\Omega_\omega)$ are unique up to unitary equivalence. If $\Ac$ is a $C^*$-algebra,
	then, additionally, (i) each $\pi_\omega(A)$ extends to a bounded operator on $\HH_\omega$; (ii)
	$\omega$ is pure if and only if the representation is irreducible; 
	and (iii) if $\pi_\omega$ is faithful [i.e., $\pi_\omega(A)=0$ only if $A=0$] then $\|\pi_\omega(A)\|=\|A\|_{\Ac}$.  
\end{theorem}
Due to the fact that $\pi_\omega(\Ac)\Omega_\omega$ is dense in $\HH_\omega$, we say that $\Omega_\omega$ is \emph{cyclic} for the representation.
The existence of a link between purity and irreducibility is easily understood from the following example: if $\psi,\varphi\in\HH$ are linearly independent (normalised) vector states on a subalgebra $\Ac$ of $\BB(\HH)$, then the density matrix state
\[
\rho = \lambda \ket{\psi}\bra{\psi} + (1-\lambda)\ket{\varphi}\bra{\varphi}\qquad
0<\lambda<1
\]
can be realised as the vector state $\Psi_\rho=\sqrt{\lambda}\psi\oplus \sqrt{1-\lambda}\varphi$ in the reducible representation
\[
A\mapsto A\oplus A
\]
of $\Ac$ on $\HH\oplus\HH$.


\begin{proof}
	The construction of the GNS representation has several steps:
	\begin{itemize}
		\item Define a subset $\Ic_\omega\subset \Ac$ by 
		\[
		\Ic_\omega = \{A\in\Ac: \omega(A^*A)=0\}.
		\]
		Using the Cauchy--Schwarz identity~\eqref{eq:CauchySchwarz}, one may prove that $\Ic_\omega$ is a left ideal: if $A\in\Ic_\omega$ then
		\[
		|\omega((BA)^*(BA))|^2 = |\omega(C^*A)|^2\le \omega(C^*C)\omega(A^*A) = 0,
		\]
		where $C=B^*BA$, 
		so $BA\in\Ic_\omega$. A similar argument shows that $\Ic_\omega$ is a subspace of $\Ac$ and that
		the subspace $\Ic_\omega^*$ is a right ideal. \emph{(Exercises!)}
		\item Define a quotient vector space
		\[
		\DD_\omega :=\Ac/\Ic_\omega
		\]
		and note that $\DD_\omega$ carries an inner product defined by 
		\[
		\ip{[A]}{[B]} = \omega(A^*B) \qquad A,B\in\Ac,
		\]
		which is well-defined (i.e., independent of the choice of representatives, and has the properties of an inner product) due to the Cauchy--Schwarz inequality and definition of $\Ic_\omega$ \emph{(Exercises!)}. For example,
		\[
		\ip{[A]}{[A]} = 0\quad\iff\quad \omega(A^*A)=0\quad\iff\quad A\in\Ic_\omega \quad\iff\quad[A]=0.
		\] 
		\item Define the Hilbert space $\HH_\omega$ as the completion of $\DD_\omega$ with respect to the inner product just mentioned; by construction $\DD_\omega$ is dense in $\HH_\omega$. 
		\item Define $\Omega_\omega = [\II]\in\DD_\omega\subset\HH_\omega$, noting that
		$\|\Omega_\omega\|^2 = \omega(\II) = 1$.
		\item Define $\pi_\omega$ as follows: for each $A\in\Ac$, $\pi_\omega(A)$ is the linear operator $\DD_\omega\to\DD_\omega$ given by 
		\[
		\pi_\omega(A)[B] := [AB],\qquad B\in\Ac
		\]
		which is well-defined due to the left-ideal property yet again.
		The properties
		\begin{align*}
		\pi_\omega(\II) &= \II_\HH   \\
		\pi_\omega(A+\lambda B) &= \pi_\omega(A)+ \lambda\pi_\omega(B) \\
		\pi_\omega(AB) &= \pi_\omega(A)\pi_\omega(B) 
		\end{align*}
		are easily verified as identities of operators on $\DD_\omega$, and
		it is clear that $\pi_\omega(\Ac)\Omega_\omega=\DD_\omega$, so $\Omega_\omega$ is cyclic as $\DD_\omega$ is dense.
		
		Now observe that 
		\[
		\ip{\pi_\omega(A^*)[B]}{[C]} = \ip{[A^*B]}{[C]} =\omega(B^*AC) = \ip{[B]}{\pi_\omega(A)[C]}
		\] 
		for all $[B],[C]\in\DD_\omega$. This shows that $\pi_\omega(A)$ has an adjoint with dense domain including $\DD_\omega$ and indeed
		\begin{align*}
		\pi_\omega(A^*) &= \pi_\omega(A)^*|_{\DD_\omega}.
		\end{align*}
		Therefore $(\HH_\omega,\DD_\omega,\pi_\omega)$ is a representation of $\Ac$, and the calculation
		\[
		\ip{\Omega_\omega}{\pi_\omega(A)\Omega_\omega} = \ip{[\II]}{\pi_\omega(A)[\II]} = \ip{[\II]}{[A]}=\omega(A)
		\]
		verifies~\eqref{eq:GNSprop}. 
		\item For uniqueness, suppose that another Hilbert space $\HH$, dense domain $\DD$, distinguished vector $\Omega$ and representation $\pi$ are given with the properties stated above. Now define a linear map $U:\DD\to \DD_\omega$ by
		\[
		U\pi(A)\Omega = \pi_\omega(A)\Omega_\omega
		\]
		noting that this is well-defined because 
		\[
		\pi(A)\Omega=0 \quad \implies \quad 0=\|\pi(A)\Omega\|_\HH^2=\omega(A^*A)=\|\pi_\omega(A)\Omega_\omega\|_{\HH_\omega}^2 \quad\implies\quad  \pi_\omega(A)\Omega_\omega = 0.
		\]
		Essentially the same calculation shows that $U$ is an isometry, and as it is clearly invertible, $U$ therefore extends to a unitary $U:\HH\to\HH_\omega$ so that $U\DD=\DD_\omega$, $U\Omega=\Omega_\omega$ and
		\[
		U\pi(A)U^{-1} = \pi_\omega(A)
		\]
		acting on any vector in $\DD_\omega$. This is the promised unitary equivalence.
		\item The special features of $C^*$-algebras are addressed e.g., in \cite[Prop.~2.3.3 and \S 2.3.3]{BratRob:vol1} .
	\end{itemize}
\end{proof}
\begin{exercise}
	Work through all the details in the proof of the GNS representation theorem.
\end{exercise}
	\section{Case study: the van Hove model}\label{sec:vH}
	
	Our aim in this section is to present a version of van Hove's model, one of the first instances in which it became clear that unitarily inequivalent representations of the CCRs arise naturally in QFT. van Hove's work provided part of the motivation for Haag's theorem on the nonexistence of the interaction picture and the subsequent development of AQFT.
	
	The van Hove model~\cite{vanHove:1952} describes the interaction of a neutral scalar field with an external source; physically, it is a slightly simplified version of Yukawa's model~\cite{Yukawa:1935} of the mediation of
	internucleon forces by pions, and correctly accounts for the Yukawa interaction. The Lagrangian is
	\[
	\mathcal{L} =\frac{1}{2}(\nabla_\mu\phi)\nabla^\mu\phi - \frac{1}{2}m^2\phi^2
	-\rho\phi
	\]
	using the ${+}{-}{-}{-}$ signature conventions,\footnote{We also adopt units in which the speed of light and $\hbar$ are both set to $1$.}
	where the source $\rho=\rho(\xb)$ is a time-independent real-valued function, or even a distribution: a simple nucleon model might have $\rho$ as a linear combination of $\delta$-functions.\footnote{The Compton wavelength of the pions is approximately seven times that of a proton.} The field equation is
	\begin{equation}\label{eq:vHeq}
	(\Box+m^2)\phi=-\rho.
	\end{equation}
	Of course, $\rho=0$ is the familiar massive real scalar field, and in general
	we know how to solve an inhomogeneous equation by finding a particular integral
	and reducing to a homogeneous equation. So this model ought to be (and is) exactly solvable. However, let us forget this for a moment and proceed according to standard lore (and in the first instance, a little formally), so as to exhibit the limitations of naive QFT. 
	
	We shall outline three quantization approaches. The first two follow the spirit of wave mechanics. Namely, having established a Hilbert space representation for the free field, we stick with it.
	\paragraph{Method 1: Schr\"odinger picture} We work in the standard Fock space of the $\rho=0$ model, in which there are `time zero' fields $\varphi(\xb)$, $\pi(\xb)$ obeying the equal time commutation relations
	\begin{equation}\label{eq:ETCR}
	[\varphi(\xb),\pi(\yb)]=i\delta(\xb-\yb)\mathbb{1}.
	\end{equation}
	To treat the `interacting' model, one constructs its Hamiltonian $H_\rho$ in terms of the  canonical variables $\varphi(\xb)$, $\pi(\xb)$, obtaining the Schr\"odinger picture evolution $e^{-iH_\rho t}$ of states. If desired, one can then pass to 
	the interacting field in the Heisenberg picture as 
	\[
	\Phi_\rho(t,\xb)=e^{iH_\rho t}\varphi(\xb)e^{-iH_\rho t}.
	\]
	\paragraph{Method 2: Interaction picture} In the interaction picture, the time-dependent
	field is given by the free field $\Phi_0(t,\xb)$ acting on the usual Fock space. The interaction picture state evolution may be obtained e.g., using a perturbative Dyson expansion.	
	
	Methods 1 and 2 work perfectly well if $\rho$ is a smooth compactly supported function. However, one encounters problems if $\rho$ has $\delta$-singularities (for UV reasons) or if either $m=0$ or $\rho\equiv 1$ (for IR reasons). The third method does not suffer these problems. It is more algebraic in nature and does not start from a prejudice about what the Hilbert space of the theory should be. Rather than treating the van Hove model as a modification of the free field, we quantize it directly.  
	
	\paragraph{Method 3: Canonical quantization from scratch} Start again
	with \emph{classical} canonical variables $\varphi_\rho$ and $\pi_\rho$. Making a Fourier analysis we may (without loss) write them in the form
	\begin{align*}
	\varphi_\rho(\xb) &= \int\frac{d^3\kb}{(2\pi)^3}\frac{1}{\sqrt{2\omega}}
	(a(\kb) + a(-\kb)^*)e^{i\kb\cdot\xb}\\
	\pi_\rho(\xb) &= -i\int\frac{d^3\kb}{(2\pi)^3}\sqrt{\frac{\omega}{2}}
	(a(\kb) - a(-\kb)^*)e^{i\kb\cdot\xb}
	\end{align*} 
	for complex coefficients $a(\kb)$; here, the free mode frequency $\omega=\sqrt{\|\kb\|^2+m^2}$ has been inserted for later convenience.
	Now promote the $a(\kb)$'s to `operators' and impose the equal time commutation relations~\eqref{eq:ETCR}, whereupon the $a(\kb)$ must obey the CCRs
	\begin{equation}\label{eq:aCCRs}
	[a(\kb),a(\kb')^*]=(2\pi)^3\delta(\kb-\kb')\mathbb{1}
	\end{equation}
	and other commutators vanishing. Note that this is an entirely algebraic procedure and that nothing has yet been said about any Hilbert space representation. So the term `operator' is a bit misleading as there is, as yet, nothing to operate on.  
	
	The next step is to introduce the Hamiltonian, which differs from the free Hamiltonian by the addition of the $\rho\phi$ terms. Written 
	in normal order, it is 
	\begin{equation}\label{eq:HvH}
	H_\rho=\int\frac{d^3\kb}{(2\pi)^3} \left(\omega a(\kb)^*a(\kb) +\frac{1}{\sqrt{2\omega}}\left(\overline{\hat{\rho}(\kb)}a(\kb) + \hat{\rho}(\kb)a(\kb)^*\right) \right) ,
	\end{equation} 
	where 
	\begin{equation}
	\hat{\rho}(\kb)=\int d^3\xb\, \rho(\xb)e^{-i\kb\cdot\xb}
	\end{equation} 
	is the Fourier transform of $\rho$.
	
	We now need to find a Hilbert space representation in which (give or take a constant) $H_\rho$ is
	self-adjoint and the CCRs~\eqref{eq:aCCRs} are valid, and ideally we would
	like to find a vacuum state for $H_\rho$. This can be accomplished easily,
	by completing the square in~\eqref{eq:HvH} to give
	\begin{equation}
	H_\rho = \int\frac{d^3\kb}{(2\pi)^3} \omega \tilde{a}(\kb)^*\tilde{a}(\kb) + E_\rho\mathbb{1},
	\end{equation}
	where
	\begin{equation}\label{eq:atilde}
	\tilde{a}(\kb) = a(\kb) + \frac{\hat{\rho}(\kb)}{\sqrt{2\omega^3}}\mathbb{1},
	\end{equation}
	and the constant $E_\rho$ is
	\begin{equation}
	E_\rho= -\int\frac{d^3\kb}{(2\pi)^3} \frac{|\hat{\rho}(\kb)|^2}{2\omega^2} = \frac{1}{2}\int d^3\xb\, d^3\yb\,\rho(\xb) V_Y(\xb-\yb)\rho(\yb) \,, 
	\end{equation} 
	in which 
	\[
	V_Y(\xb)=-\frac{e^{-m\|\xb\|}}{4\pi\|\xb\|^2}
	\]
	is the Yukawa potential, responsible for the inter-nucleon force in this model. An important point is that the $\tilde{a}(\kb)$ operators clearly obey
	the same CCRs as the $a(\kb)$.
	
	It is now clear how to proceed: we should represent the $\tilde{a}(\kb)$'s and $\tilde{a}(\kb)^*$'s as Fock space annihilation and creation operators, 
	writing $\tilde{\Omega}$ for the `vacuum vector' annihilated by the $\tilde{a}(\kb)$'s. 
	Then $\tilde{\Omega}$ is automatically a state of lowest energy for $H_\rho$ and,
	by discarding the constant, we may arrange that it is a state of zero energy.
	Of course, the Fock space is just the usual bosonic Fock space over a $1$-particle space $L^2(\RR^3,d^3\kb)$ with $\tilde{\Omega}$ as the vacuum vector, 
	and our Hamiltonian is just the standard Hamiltonian of the free field on this space,
	\[
	\tilde{H}_0=\int\frac{d^3\kb}{(2\pi)^3} \omega \tilde{a}(\kb)^*\tilde{a}(\kb) .
	\] 
	It may seem that this simply reproduces the situation of methods 1 or 2. However, the redefinition~\eqref{eq:atilde} provides the crucial difference, because the time-zero fields here are defined in terms of the original operators $a(\kb)$. Indeed, a calculation shows that 
	the time-dependent Heisenberg picture field is
	\begin{equation}
	\tilde{\Phi}_\rho(t,\xb) = e^{i\tilde{H}_0 t}\varphi_\rho(\xb)e^{-i\tilde{H}_0 t} = \tilde{\Phi}_0(t,\xb) + \phi_\rho(\xb)\mathbb{1},
	\end{equation}
	where 
	\begin{equation}
	\phi_\rho(\xb) = -\int\frac{d^3\kb}{(2\pi)^3}\frac{\hat{\rho}(\kb)}{\omega^2} e^{i\kb\cdot\xb} = -(V_Y\star\rho)(\xb)
	\end{equation}
	is (formally) a time-independent solution to the field equation~\eqref{eq:vHeq}, and
	\begin{equation}
	\tilde{\Phi}_0(t,\xb) = \int\frac{d^3\kb}{(2\pi)^3}\frac{1}{\sqrt{2\omega}}
	(\tilde{a}(\kb)e^{-i\omega t + i\kb\cdot\xb} + \tilde{a}(\kb)^*e^{i\omega t - i\kb\cdot\xb})
	\end{equation} 
	is the standard free real scalar field on the Fock space.
	
	On reflection, we see that the canonical quantization procedure amounts precisely to solving the inhomogeneous equation using a particular integral and a solution to the homogeneous equation, which we then quantize as
	a free scalar field with $\tilde{\Omega}$ as its vacuum vector.	
		
	Provided that $\phi_\rho$ is at least a weak solution to~\eqref{eq:vHeq}, the final quantized model will be well-defined, in the sense that the fields give operator-valued distributions weakly solving the field equations and obeying the covariant commutation relations [which are actually the same as for the free field, as can be verified easily using Peierls method]. In particular, for $m>0$, $\rho$ could be any tempered distribution, so linear combinations of $\delta$-functions, or $\rho\equiv 1$ are certainly permitted.
			
	We can now discuss the relation of this model to the approach in Method~1,
	which started from the assumption that the time-zero fields of the van Hove model are exactly those of the free field. Our construction has actually shown that
	\begin{equation}
	\varphi_\rho(\xb) = \varphi_0(\xb) + \phi_\rho(\xb)\mathbb{1}, \qquad\pi_\rho(\xb) =
	\pi_0(\xb) 
	\end{equation}
	where $\varphi_0$ and $\pi_0$ are standard free-field time-zero fields in the vacuum representation. So methods 1 and 3 are equivalent if and only if there is a unitary
	map $U$ on the free field Fock space so that $\varphi_\rho(\xb)=U\varphi_0(\xb)U^{-1}$, $\pi_\rho(\xb)=U\pi_0(\xb)U^{-1}$, 
	or equivalently, $a(\kb)  = U \tilde{a}(\kb) U^{-1}$.
	
	As the following exercise shows, a necessary (and in fact sufficient) condition for the existence of $U$ is that 
	\begin{equation}\label{eq:vHint}
	\int \frac{d^3\kb}{(2\pi)^3} \frac{|\hat{\rho}(\kb)|^2}{2\omega^3} <\infty,
	\end{equation}
	which fails for both $\delta$-function singularities and the limit $\rho\to 1$.  Whenever the integral in~\eqref{eq:vHint} diverges, method~1, and similarly method~2, therefore fail to reproduce the exact solution to the van Hove model.
	
	\begin{exercise}
		Supposing that $U$ exists, deduce that $\psi = U\tilde{\Omega}$ is annihilated by all $a(\kb)$. Show that the $n$-particle wavefunction component of $\psi$ is
		\[
		\psi_n = \frac{1}{\sqrt{n}}\chi_\rho\otimes_s \psi_{n-1},\qquad \chi_\rho(\kb)= \frac{\hat{\rho}(\kb)}{\sqrt{2\omega^3}},
		\]
		which implies that $\psi=0$ unless $\chi_\rho\in L^2(\RR^3,d^3\kb/(2\pi)^3)$, i.e., only if~\eqref{eq:vHint} holds. If it does hold, show further that
		\[
		\psi = \mathcal{N} \sum_{n=0}^\infty \frac{\chi_\rho^{\otimes n}}{\sqrt{n!}}
		\]
		for some nonzero constant $\mathcal{N}\in\CC$,
		and hence $\|\psi\| = |\mathcal{N}| e^{\|\chi_\rho\|^2/2}$. 
	\end{exercise}
	
	The lessons to be drawn from the van Hove model are:
	\begin{itemize}
		\item Unitarily inequivalent representations of the CCRs appear naturally in QFT, even in simple models. In fact, were we to replace $\rho\mapsto  \lambda\rho$, we would find that the corresponding CCR representations are unitarily inequivalent for any distinct values of $\lambda$, when the integral in~\eqref{eq:vHint} diverges.
		\item The interaction picture does not exist in general.  \emph{Haag's theorem}, a general result proved partly in response to the difficulties pointed out by van Hove, shows that the interaction picture is of very limited applicability: any QFT on Minkowski space whose time-zero fields coincide with those of a free field theory is necessarily a free field theory. It is sometimes said that Haag's theorem applies only to translationally invariant theories; this is true of the formal statement, but the van Hove model with $\delta$-function potentials shows that the interaction picture can fail in nontranslationally invariant situations as well, or even in a finite volume box as in van Hove's original paper.
		\item The success of Method 3 may be attributed to the way it separates the problem of determining the algebraic relations, treated first, from the problem of finding a Hilbert space representation, treated second. 
	\end{itemize}

	\section{Algebraic QFT}\label{sec:AQFT}
	The appearance of unitarily inequivalent representations of the CCRs even in simple QFT models motivates an approach based in the first instance on algebraic relations. In this section we set out some minimal axioms for AQFT and then give constructions of some simple free QFTs in the AQFT framework and describe the class of quasifree states for the free scalar field. 

	\subsection{Basic requirements}
	The minimum requirements (more or less) of AQFT on Minkowski spacetime $\MM$ are:
	\begin{enumerate}[label=\bf A\arabic{enumi},leftmargin=*,widest=4] 
		\item\label{it:loc} {\bf Local algebras}\\
		There is a unital $*$-algebra $\Ac(\MM)$ and, 
		to each open causally convex\footnote{A subset is causally convex if it contains every causal curve whose endpoints belong to it; see Fig.~\ref{fig:spacetimeregions} for an illustration.} bounded region $\Ocal\subset\MM$, a subalgebra $\Ac(\Ocal)$ containing the unit of $\Ac(\MM)$, so that the $\Ac(\Ocal)$ collectively generate $\Ac(\MM)$. We call $\Ac(\MM)$ the \emph{quasi-local algebra} of the theory.
		\item\label{it:iso} {\bf Isotony}\footnote{Doubtless to avoid `monotony'.}\\
		Whenever $\Ocal_1\subset \Ocal_2$, the corresponding local algebras are nested, 
		\begin{equation}
		 \Ac(\Ocal_1)\subset \Ac(\Ocal_2).
		\end{equation} 
		\item\label{it:Eins} {\bf Einstein causality}\\ If $\Ocal_1$ and $\Ocal_2$ are causally disjoint then
		\begin{equation}
		[\Ac(\Ocal_1),\Ac(\Ocal_2)] = 0,
		\end{equation}
		i.e., 
		\begin{equation}
			[ A_1, A_2] = 0,\qquad\text{for all $A_i\in\Ac(\Ocal_i)$, $i=1,2$}.
		\end{equation} 
		\item\label{it:Poinc} {\bf Poincar\'e covariance}\\
		To every transformation $\rho$ in the identity connected component $\Pc_0$ of the Poincar\'e group, there is an automorphism $\alpha(\rho)$ of $\Ac(\MM)$ such that
		\begin{equation}
		\alpha(\rho):\Ac(\Ocal)\to \Ac(\rho \Ocal)
		\end{equation}
		such that $\alpha(\id)=\id_{\Ac(\MM)}$
		and naturally $\alpha(\sigma) \circ\alpha(\rho)  = \alpha(\sigma\circ\rho)$ is required for any  $\sigma,\rho\in\Pc_0$.  
		\item\label{it:dyn} {\bf Existence of dynamics}\\
		If $\Ocal_1\subset \Ocal_2$ and $\Ocal_1$ contains a Cauchy surface of $\Ocal_2$,\footnote{Recall that a Cauchy surface for $\Ocal_2$ is a subset met exactly once by every inextendible timelike curve in $\Ocal_2$.} then 
		\begin{equation}
		\Ac(\Ocal_2) = \Ac(\Ocal_1).
		\end{equation}
	 \emph{Sometimes this condition is weakened.}
	\end{enumerate}
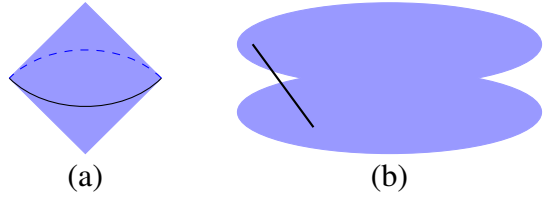
\begin{figure}
	\begin{center}
		\begin{tikzpicture}
		\filldraw[fill=blue!40!white,,color=blue!40!white](0,0) node[below,color=black]{(a)} --++(1,1)--++(-1,1)--++(-1,-1)--cycle;
		\draw[black] (1,1) .. controls (0.5,0.5) and (-0.5,0.5) .. (-1,1);
		\draw[dashed, blue] (1,1) .. controls (0.5,1.5) and (-0.5,1.5) .. (-1,1);
		\node[below] at (4,0) {(b)};
		\draw[fill=blue!40!white,color=blue!40!white] (4,1.45) ellipse (2 and 0.55);
		\draw[fill=blue!40!white,color=blue!40!white] (4,0.55) ellipse (2 and 0.55);
		\draw[thick] (2.2,1.45) -- ++ (0.8,-1.1);
		\end{tikzpicture}
	\end{center}\caption{Sketches of two spacetime regions. (a) A double cone region, which is an example of a causally convex set; (b) illustrates a non-causally convex set because points in the region may be joined by a causal curve (e.g., the black line) that leaves it.
		Sketch of a double cone region.}\label{fig:spacetimeregions}
\end{figure}
	The interpretation of this structure is that the self-adjoint elements of $\Ac(\Ocal)$,
	i.e., $A=A^*\in\Ac(\Ocal)$ are observables that are associated with the region $\Ocal$.
	Loosely, we may think of them as being measureable within $\Ocal$; this can be made more precise, see \cite{FV:2018} and the short account~\cite{Fewster:2019}. We will use the term \emph{local observable algebra} to describe $\Ac(\Ocal)$.
	The set of open causally convex bounded regions is a \emph{directed system} under inclusion: given any two regions $\Ocal_1$ and $\Ocal_2$, there exists a further region $\Ocal_3$ containing them both. For this reason the assignment $\Ocal\mapsto\Ac(\Ocal)$ is often called a \emph{net} of local algebras. In this situation, the condition that the $\Ac(\Ocal)$ collectively generate $\Ac(\MM)$ reduces to the issue of whether $\Ac(\MM)$ is equal to the union $\bigcup_{\Ocal_n}\Ac(\Ocal_n)$ [or its closure, if the algebras carry suitable topologies; see below], where $\Ocal_n$ is any nested family of regions with $\bigcup_n \Ocal_n=\MM$.
	
	We will add more conditions in due course.
	\paragraph{Remarks}
	\begin{itemize}
	\item The prototypical local region is a \emph{double cone}:\index{double cone} namely, the set of all points
	lying on smooth timelike curves between two points $p$ and $q$ (with $p$ and $q$ excluded). By scalings, boosts and translations, all double cones in Minkowski space can be obtained from the elementary example
	\[
	\Ocal = \{(t,\xb)\in\RR^4: |t|+\|\xb\| < \ell_0 \}
	\]
	for some length-scale $\ell_0>0$; see Fig.~\ref{fig:spacetimeregions}. All double cones are causally convex.   
	\item For technical reasons it is often useful to require that each $\Ac(\Ocal)$ is in fact a $C^*$-algebra, but we need not insist on this, nor even that the $\Ac(\Ocal)$ carry any topology. In the $C^*$-case, one would require $\Ac(\MM)$ to be generated in a $C^*$-sense by the local algebras -- technically it is their $C^*$-inductive limit. In particular, the union $\bigcup_{\Ocal}\Ac(\Ocal)$ would be dense in $\Ac(\MM)$. 	
	\item Sometimes (particularly in curved spacetimes) it is convenient to allow for local algebras indexed by unbounded ($=$ noncompact closure) regions. 
 
	\item Einstein causality requires elements of algebras of spacelike separated regions to commute. Therefore Fermi fields can only appear in products involving even numbers of fields. We return to this later.	
	\item As mentioned, these are minimal requirements for AQFT but do not, by themselves, suffice to distinguish a quantum field theory from other relativistic models. 
	\item Nothing has yet been said about Hilbert spaces, or about what algebraic states on the observable algebras are to be regarded as physical; we will discuss these issues later. Note that one can do quite a lot without ever going into Hilbert spaces. For example, 
		let $\Ac(\MM)$ be the algebra of the free real scalar field [see below], and let $\omega$ be a state on $\Ac(\MM)$. Then the smeared $n$-point function is
		\[
		W_n(f_1,f_2,\ldots,f_n) := \omega(\Phi(f_1)\Phi(f_2)\cdots \Phi(f_n)), \qquad f_1,\ldots,f_n\in\CoinX{\MM}
		\]
		and if this is suitably continuous w.r.t.\ the $f_i$, it defines a distribution $W_n\in\DD'(\MM^n)$.
		Here, $\Phi(f)$ denotes a `smeared field' as will be described shortly. 
		Therefore sufficiently regular states $\omega$ define a hierarchy of distributional $n$-point functions, without ever using a Hilbert space. Here, as elsewhere in these notes, $\CoinX{\MM}$ denotes the space of smooth, complex-valued functions on $\MM$ with compact support (i.e., vanishing outside a bounded set).
	\end{itemize}
	
	\subsection{Examples}\label{sec:examples}
	
	We continue by giving some specific examples of field theories in the framework of AQFT, drawing out various lessons as we go. 
	
	\paragraph{Free real scalar field} Consider the field equation
	\begin{equation}
	(\Box+ m^2)\phi = 0
	\end{equation}
	and let $E^+$ and $E^-$ be the corresponding retarded and advanced Green operators, i.e., 
	$\phi = E^\pm f$ solves the inhomogeneous equation
	\begin{equation}
	(\Box+ m^2)\phi = f
	\end{equation}
	and the support of $\phi$ is contained in the causal future ($+$, retarded) or causal past ($-$, advanced) of the support of $f$. Also define
	the advanced-minus-retarded solution operator $E=E^--E^+$ and write
	\begin{equation}
	E(f,g)= \int_\MM f(x) (Eg)(x)\,\dvol_\MM(x)\,,
	\end{equation}
	where $\dvol_\MM(x)\equiv d^4x$.
	The integral kernel is familiar from standard QFT:
	\begin{equation}
	E(x,y) = -\int \frac{d^3\kb}{(2\pi)^3} \frac{\sin k\cdot (x-y)}{\omega},
	\end{equation}
	where $k^\bullet=(\omega,\kb)$, $\omega=\sqrt{\|\kb\|^2+m^2}$, and we use standard inertial coordinates on Minkowski spacetime.
	\begin{exercise}\label{ex:preETCR}
		Show that $E(x,y)|_{y^0=x^0}=0$, and $\partial_{y^0}E(x,y)|_{y^0=x^0}=\delta(\xb-\yb)$.
	\end{exercise}
	
	To formulate the quantized field in AQFT, we give generators and relations for the desired algebra $\Ac(\MM)$, thus specifying it uniquely up to isomorphism. For completeness, a detailed description of the construction is given in Appendix~\ref{appx:presentation}. The generators are written $\Phi(f)$, labelled by test functions $f\in \CoinX{\MM}$ -- for the moment this just a convenient way of writing them; there is no underlying field $\Phi(x)$ to be understood here. The relations imposed, for all test functions $f,g\in\CoinX{\MM}$, are:
	\begin{enumerate}[label=\bf SF\arabic{enumi},leftmargin=*,widest=4] 
		\item\label{it:SFfieldlin} {\bf Linearity} $f\mapsto \Phi(f)$ is complex linear 
		\item {\bf Hermiticity} \label{it:SFfieldcong} $\Phi(f)^*=\Phi(\overline{f})$
		\item\label{it:SFfieldeq} {\bf Field equation}  $\Phi((\Box+ m^2)f)= 0$
		\item\label{it:SFfieldcom} {\bf Covariant commutation relations} $[\Phi(f),\Phi(g)] = iE (f,g)\II$.
	\end{enumerate}
As a consequence of the identities in Exercise~\ref{ex:preETCR},~\ref{it:SFfieldcom} is a covariant form of the equal time commutation relations; a nice way of seeing this directly is to follow the on-shell Peierls' method~\cite{Peierls:1952} to find a covariant Poisson bracket for the classical theory. The axioms we have just stated may be regarded as the result of applying Dirac's quantisation rule to this bracket, with $\Phi(f)$ regarded as the quantization of the observable
	\begin{equation}
	F_f[\phi]=\int f\phi\,\dvol_\MM
	\end{equation} 
	on a suitable solution space to the Klein--Gordon equation.  
	
	One should check, of course, that the algebra $\Ac(\MM)$ is nontrivial. It is not too hard [though we will not do this here] to show that the underlying vector space of $\Ac(\MM)$ is isomorphic to the symmetric tensor vector space
	\[
	\CC \oplus \bigoplus_{n=1}^\infty Q^{\odot n}, \qquad Q = \CoinX{\MM}/P \CoinX{\MM},
	\]
	where $P=\Box+m^2$ and $\odot$ denotes a symmetrised tensor product.  
	Therefore the nontriviality of $\Ac(\MM)$ reduces to the question of whether the quotient space $Q$ is nontrivial. 
	The latter follows from the properties of the Green operators, which can be summarised in an exact sequence~\cite{BarGinouxPfaffle}
	\begin{equation}\label{eq:exactsequence}
	0 \longrightarrow \CoinX{\MM}\stackrel{P}{\longrightarrow}  \CoinX{\MM} \stackrel{E}{\longrightarrow}
	C^\infty_{\textit{sc}} (\MM) \stackrel{P}{\longrightarrow}C^\infty_{\textit{sc}} (\MM) \longrightarrow 0,
	\end{equation} 
	where the subscript $sc$ denotes a space of functions with `spatially compact' support, which means that they vanish in the causal complement of a compact set. Together with the isomorphism theorems for vector spaces, this gives
	\begin{equation}
	Q =\CoinX{\MM}/\im P = \CoinX{\MM}/\ker E \cong \im E = \ker P =\{\phi\in C^\infty_{\textit{sc}} (\MM): P\phi=0\}=:\Sol(\MM).
	\end{equation}
	Thus $Q$ is isomorphic to the space of smooth Klein--Gordon solutions with spatially compact support, and is therefore nontrivial. Consequently, $\Ac(\MM)$ is a nontrivial unital $*$-algebra.
	
	Now define, for each causally convex open bounded $\Ocal\subset\MM$, the algebra $\Ac(\Ocal)$ to be the subalgebra of $\Ac(\MM)$ generated by elements $\Phi(f)$ for $f\in\CoinX{\Ocal}$, along with the unit $\II$. 
	Then it is clear that, if $\Ocal_1\subset \Ocal_2$, then $\Ac(\Ocal_1)\subset \Ac(\Ocal_2)$. Then properties~\ref{it:loc},~\ref{it:iso} are automatic.
	Next, because $E(f,g)=0$ when $f$ and $g$ have causally disjoint support (as $Eg$ is supported in the union of the causal future and past of $\supp g$) it is clear that all generators of $\Ac(\Ocal_1)$ commute with all generators of $\Ac(\Ocal_2)$; hence~\ref{it:Eins} holds. Next, let $\rho\in\Pc_0$. Then the Poincar\'e covariance of $\Box+m^2$ and $E$ can be used to show that
	the map of generators $\alpha(\rho)\Phi(f) = \Phi(\rho_* f)$, $(\rho_*f)(x)=f(\rho^{-1}(x))$, is compatible with the relations and extends to a well-defined unit-preserving $*$-isomorphism
	\begin{equation}
	\alpha(\rho):\Ac(\MM)\to\Ac(\MM)\,.
	\end{equation}
	Clearly $\alpha(\rho)$ maps each $\Ac(\Ocal)$ to $\Ac(\rho\Ocal)$; as we also have $\alpha(\sigma)\circ\alpha(\rho)=\alpha(\sigma\circ\rho)$,  condition~\ref{it:Poinc} holds. 
	
	Finally, let $\Ocal_1\subset \Ocal_2$ such that $\Ocal_1$ contains a Cauchy surface of $\Ocal_2$. Then any solution $\phi=Ef_2$ for $f_2\in\CoinX{\Ocal_2}$ can be written as
	$\phi=Ef_1$ for some $f_1\in\CoinX{\Ocal_1}$. An explicit formula is 
	\[
	f_1 = P \chi \phi
	\]
	where $\chi\in C^\infty(\Ocal_2)$ vanishes to the future of one Cauchy surface of $\Ocal_2$ contained in $\Ocal_1$ and equals unity to the past of another (since $\Ocal_1$ contains a Cauchy surface of $\Ocal_2$, it actually contains many of them).
	Then $\Phi(f_2)=\Phi(f_1)$, which implies that $\Ac(\Ocal_2)=\Ac(\Ocal_1)$ as required by~\ref{it:dyn}. 
	
	In fact this whole construction can be adapted to any globally hyperbolic spacetime, which is the setting in which~\eqref{eq:exactsequence} was proved~\cite{BarGinouxPfaffle}. Here, we recall that a globally hyperbolic spacetime is a time-oriented Lorentzian spacetime containing a Cauchy surface. 
	
	\paragraph{Real scalar field with external source (van Hove encore)} Let $\rho\in\DD'(\MM)$ be a distribution that is real in the sense $\overline{\rho(f)}=\rho(\overline{f})$, and let $\phi_\rho\in\DD'(\MM)$ be any weak solution to $(\Box+m^2)\phi_\rho= -\rho$. 
	
	The AQFT formulation of the real scalar field with external source $\rho$ is given in terms of \emph{the same algebras} $\Ac(\Ocal)$ as in the homogeneous case. The only difference is that we define smeared fields 
	\[
	\Phi_\rho(f) =\Phi(f) +\phi_\rho(f)\II ,
	\]
	where $\Phi(f)$ are the generators used to construct $\Ac(\MM)$ by~\ref{it:SFfieldlin}--\ref{it:SFfieldcom},
	and observe that they obey the algebraic relations 
		\begin{enumerate}[label=\bf vH\arabic{enumi},leftmargin=*,widest=4] 
		\item $f\mapsto \Phi_\rho(f)$ is complex linear 
		\item $\Phi_\rho(f)^*=\Phi_\rho(\overline{f})$ 
		\item $\Phi_\rho((\Box+ m^2)f)+\rho(f)\II=0$ 
		\item $[\Phi_\rho(f),\Phi_\rho(g)] = iE (f,g)\II$ 
	\end{enumerate}
	which are the relations that would be obtained from Dirac quantisation of the classical theory. 
	Two remarks are in order:
	\begin{itemize}
		\item We see that the algebra is not so specific to the theory in hand -- this is a general feature of AQFT: what is more interesting is how the local algebras fit together and how the elements can be labelled by fields.
		\item All the difficulties encountered in section~\ref{sec:vH} appear to have vanished. Actually, they have been moved into the question of the unitary (in)equivalence of
		representations of the algebra $\Ac(\MM)$. The separation between the algebra and its representations makes for a clean conceptual viewpoint. 
	\end{itemize}

	\paragraph{Weyl algebra} We return to the real scalar field and note two formal identities: first, that
	\begin{equation}\label{r1}
	(e^{i\Phi(f)})^* = e^{-i\Phi(f)} = e^{i\Phi(-f)}
	\end{equation}
	if $f$ is real-valued; second, from the 
	Baker--Campbell--Hausdorff formula, that
	\begin{equation}\label{r2}
	e^{i\Phi(f)}e^{i\Phi(g)} = e^{i\Phi(f)+i\Phi(g) - [\Phi(f),\Phi(g)]/2} = 
	e^{-iE(f,g)/2}e^{i\Phi(f+g)}.
	\end{equation} 
 As there is no topology on $\Ac(\MM)$ we cannot address any convergence questions, so these are to be understood as identities between formal power series in $f$ and $g$.
	We may also note that $\Phi(f)$ and $E(f,g)$ depend only on the equivalence classes of $f$ and $g$ in $\CoinX{\MM}/P\CoinX{\MM}\cong\Sol(\MM)$. Moreover, the exponent in~\eqref{r2} is related to the symplectic product on the space of real-valued Klein--Gordon solutions, $\Sol_\RR(\MM)\cong\CoinX{\MM;\RR}/P\CoinX{\MM;\RR}$ by
	\begin{equation}\label{eq:symp}
	\sigma([f],[g])= E(f,g).
	\end{equation}
	
	These considerations motivate the definition of a unital $*$-algebra, generated by symbols $\Ws([f])$ labelled by $[f]\in \CoinX{\MM;\RR}/P\CoinX{\MM;\RR}$ and satisfying relations mimicking \eqref{r1} and \eqref{r2}. In fact, any real symplectic space $(S,\sigma)$ determines a unital $*$-algebra, generated by symbols $\Ws(\phi)$, $\phi\in S$ and satisfying the relations:		
	\begin{enumerate}[label=\bf W\arabic{enumi},leftmargin=*,widest=2] 
		\item $\Ws(\phi)^*=\Ws(-\phi)$
		\item $\Ws(\phi)\Ws(\phi') = e^{-i\sigma(\phi,\phi')/2}\Ws(\phi+\phi')$.
	\end{enumerate}
It is a remarkable fact that this algebra can be given a $C^*$-norm and completed
to form a $C^*$-algebra in exactly one way (up to isomorphism)~\cite{BratRob:vol2}. This is the \emph{Weyl algebra} $\Wc(S,\sigma)$. In our case of interest $S=\CoinX{\MM;\RR}/P\CoinX{\MM;\RR}$, the symplectic form is given by~\eqref{eq:symp}, and we will denote the corresponding Weyl algebra by $\Wc(\MM)$ for short. As before, we can form local algebras by defining $\Wc(\Ocal)$ as the $C^*$-subalgebra generated by $\Ws([f])$'s with $\supp f\subset \Ocal$ and $\Ocal$ being any causally convex open bounded subset of $\MM$. 
\begin{exercise} 
		In a general Weyl algebra $\Wc(S,\sigma)$, prove that $\Ws(0)=\II$, the algebra unit.
\end{exercise}  
	
	It is worth pausing to examine the explicit construction of the Weyl algebra $\Wc(S,\sigma)$.  Consider
	the (inseparable) Hilbert space $\HH=\ell^2(S)$
	of square-summable sequences $a=(a_\phi)$ indexed by $\phi\in S$, and define
	\begin{equation}\label{eq:Weyl_conc}
	(\Ws(\phi')a)_\phi = e^{-i\sigma(\phi',\phi)/2} a_{\phi+\phi'}.
	\end{equation}
	Obviously the $\Ws(\phi)$'s are all unitary. 
	The Weyl algebra $\Wc(S,\sigma)$ is the closure of the $*$-algebra generated by the $\Ws(\phi)$'s in the norm topology on $\BB(\HH)$, equipped with the operator norm.

	\begin{exercise}
		Check that~\eqref{eq:Weyl_conc} implies $\Ws(\phi)=\Ws(-\phi)^*$ and $\Ws(\phi)\Ws(\phi') = e^{-i\sigma(\phi,\phi')/2}\Ws(\phi+\phi')$.
	\end{exercise} 
	\begin{exercise}
		Let $\Omega\in\ell^2(S)$ be the sequence $(\delta_{\phi,0})$, where
		\[
		\delta_{\phi,0}= \begin{cases}
		1 & \phi=0\\ 0 & \phi\neq 0.
		\end{cases} 
		\]
		If $\phi\neq\phi'$, show that
		$\Ws(\phi)\Omega$ and $\Ws(\phi')\Omega$ are orthogonal and deduce that 
		\[
		\|\Ws(\phi)-\Ws(\phi')\| =2.
		\]
		This shows that there are no nonconstant continuous curves in the Weyl algebra.
	\end{exercise}
	A corollary of the exercise is that one cannot differentiate $\lambda\mapsto \Ws([\lambda f])$ \emph{within the Weyl algebra} in the hope of recovering a smeared field operator, nor can we exponentiate $i\Phi(f)$ \emph{within the algebra $\Ac(\MM)$} to obtain a Weyl operator. The heuristic relationship 
	$\Ws([f]) = e^{i\Phi(f)}$ does not hold literally in either of these algebras. 
\begin{exercise}
		Show that the GNS representation of the (abstract) Weyl algebra over symplectic space $(S,\sigma)$ induced by the \emph{tracial state} $\omega_{\text{tr}}(\Ws(\phi)) = \delta_{\phi,0}$ 
		coincides with the concrete construction of a representation on $\Hcal=\ell^2(S)$
		given earlier, with the GNS vector   $\Omega_{\text{tr}}=(\delta_{\phi,0})$,  i.e., $\omega_{\text{tr}}(A)=\ip{\Omega_{\text{tr}}}{A\Omega_{\text{tr}}}$. 
	\end{exercise}
	
	\paragraph{Complex scalar field} The algebra of the free complex scalar field $\Cc(\MM)$
	may be generated by symbols $\Phi(f)$ ($f\in\CoinX{\MM}$) subject to the relations
	\begin{enumerate}[label=\bf CF\arabic{enumi},leftmargin=*,widest=4] 
		\item\label{it:CFfieldcong} {\bf Linearity} $f\mapsto \Phi(f)$ is complex linear 
		\item\label{it:CFfieldeq} {\bf Field equation} $\Phi((\Box+ m^2)f)= 0$ 
		\item\label{it:CFfieldcom} {\bf Covariant commutation relations} $[\Phi(f),\Phi(g)]=0$ and $[\Phi(f)^*,\Phi(g)] = iE (\bar{f},g)\II$.
	\end{enumerate}
	It is also usual to write $\Phi^\star(f):=\Phi(\bar{f})^*$, so that $f\mapsto\Phi^\star(f)$ is  also complex-linear. This algebra admits a family of automorphisms $\eta_\alpha$ given on the generators by 
	\begin{equation}\label{eq:gaugeCSF}
	\eta_\alpha(\Phi(f))= e^{-i\alpha} \Phi(f),
	\end{equation}
	corresponding to a global $U(1)$ gauge symmetry of the classical complex field.  
	\begin{exercise}
	Check that~\eqref{eq:gaugeCSF} extends to a well-defined automorphism for each $\alpha\in\RR$. Show also that there is an isomorphism between $\Cc(\MM)$ and
	the algebraic tensor product $\Ac(\MM)\otimes\Ac(\MM)$ of two copies of the real scalar field algebra (with the same mass $m$) given on generators by  
	\[
	\Phi(f)\mapsto \frac{1}{\sqrt{2}}\left(\Phi_r(f)\otimes \II+i \II\otimes \Phi_r(f) \right), \qquad f\in\CoinX{\MM},
	\]
	where we temporarily use $\Phi_r(f)$ to denote the generators of $\Ac(\MM)$. In this sense, the complex field is simply two independent real scalar fields.	
	\end{exercise}
	We may identify local algebras $\Cc(\Ocal)$ in the same way as before, and within each of these, identify the subalgebra $\Cc_{\text{obs}}(\Ocal)$ consisting of all elements
	of $\Cc(\Ocal)$ that are invariant under the global $U(1)$ gauge action. These are the local observable algebras. 
	
	The reader may notice that the real scalar field admits a global $\ZZ_2$ gauge symmetry generated by $\Phi(f)\mapsto -\Phi(f)$, and that the theory of two independent real scalar fields with the same mass admits an $O(2)$ gauge symmetry, of which the $U(1)$ symmetry corresponds to the $SO(2)$ subgroup. Why, then, do we not restrict the observable algebra of the real scalar field, or further restrict the observable algebras of the complex field? The answer is simply that these are physical choices. The $U(1)$ gauge invariance of the complex scalar field is related to charge conservation, while the additional $\ZZ_2$ symmetry
	in the isomorphism $O(2)\cong U(1)\rtimes \ZZ_2$ 
	corresponds to charge reversal. If the theory is used to model a charge that is conserved in nature, but for which states of opposite charge can be physically distinguished, then the correct approach is to proceed as we have done. 
	
	\paragraph{Dirac field} One can proceed in a similar way to define an algebra $\Fc(\MM)$ with generators $\Psi(u)$ and $\Psi^+(v)$ labelled by cospinor and spinor test 
	functions respectively, and with relations abstracted from standard QFT (this is left as an exercise). However, the resulting local algebras $\Fc(\Ocal)$ do not obey Einstein causality -- if $u$ and $v$ are spacelike separated then $\Psi(u)$ and $\Psi^+(v)$ anticommute. Of course we do not expect to be able to measure a smeared spinor field by itself, and what one can do instead is consider algebras generated by second degree products of the spinor and cospinor fields, labelled by (co)spinor test functions in supported in $\Ocal$. The resulting algebras $\Ac(\Ocal)$ then obey the axioms \ref{it:loc}--\ref{it:dyn}. 
	
	Meanwhile, the algebras $\Fc(\Ocal)$ obey~\ref{it:loc},~\ref{it:iso},~\ref{it:Poinc},~\ref{it:dyn} and a graded version of~\ref{it:Eins}. We describe them as constituting \emph{local field algebras}, to emphasise the fact that they contain elements carrying the
	interpretation of smeared unobservable fields. We will come back to the discussion of $\Fc(\Ocal)$ and their relation to $\Ac(\Ocal)$ in section~\ref{sec:sectors} on superselection sectors.  
	

%

\subsection{Quasifree states for the free scalar field}\label{sec:quasifree} 
 
We have seen how local algebras for the free scalar field may be constructed.  
As emphasised in Sec.~\ref{sec:AQM}, however, this is only half of the data needed for
a physical theory: we also need some states, and (for many purposes) the corresponding 
GNS representations. These are nontrivial problems in general: one needs to fix the value of $\omega(A)$ and verify the positivity
condition that $\omega(A^*A)\ge 0$ for every element $A\in\Ac(\MM)$;
furthermore, while the GNS representation is fairly explicit, it evidently involves a lot of work to do it by hand. 

Fortunately, in the case of free fields, there is a special family of \emph{quasifree states}\index{state!quasifree} where quite explicit constructions can be given. In particular, these states are determined by their two-point functions and all the conditions required of a state can be expressed in those terms. Moreover, the eventual GNS Hilbert space is a Fock space, and the smeared field operators in the representation may be given by explicit formulae. It is important to note that these include representations that are unitarily inequivalent to the representation built on the standard vacuum state. Once again, many of the arguments we will use carry over directly to curved spacetimes.

Let $W$ be any bilinear form on $\CoinX{\MM}$ obeying
\begin{equation}\label{eq:WE}
W(f,g) - W(g,f) = i E(f,g), \qquad \forall f,g\in\CoinX{\MM}
\end{equation}
and which induces a positive semidefinite sesquilinear form on $\Sol(\MM)$ by the formula 
\[
w(Ef,Eg) = W(\overline{f},g), \qquad \forall f,g\in\CoinX{\MM}.
\]  
In particular, $W$ must be a weak bisolution to the Klein--Gordon equation so that $w$ is well-defined. 
\begin{exercise}
Check that the positivity condition $W(\bar{f},f)\ge 0$ implies that 
$\overline{W(f,g)}=W(\bar{g},\bar{f})$. Also derive the Cauchy--Schwarz inequality
\begin{equation}\label{eq:CSw}
|\Im w(\phi,\phi')|^2 \le w(\phi,\phi)w(\phi',\phi'), \qquad \phi,\phi'\in\Sol(\MM).
\end{equation}
\end{exercise}

Under the above conditions, it may be proved (cf.~Prop~3.1 in~\cite{KayWald:1991}) that there is a complex Hilbert space $\HH$ and a real-linear map $K:\Sol_\RR(\MM)\to\HH$
so that $K\Sol_\RR(\MM)+iK\Sol_\RR(\MM)$ is dense in $\HH$ and
\begin{equation}\label{eq:1PS}
\ip{KEf}{KEg}_\HH = W(f,g),\qquad f,g\in\CoinX{\MM;\RR}.
\end{equation}
(These structures are unique up to unitary equivalence.)

The full construction is essentially explicit, but slightly involved. However, there is a particularly simple and interesting case, arising when the Cauchy--Schwarz inequality~\eqref{eq:CSw} is saturated, i.e.,
\begin{equation}\label{eq:saturation}
\sup_{\phi'\neq 0}\frac{|\Im w(\phi,\phi')|^2}{w(\phi,\phi)w(\phi',\phi')} = 1 \qquad\text{for all $0\neq \phi\in\Sol_{\RR}(\MM)$},
\end{equation}
where the supremum is taken over nonzero elements of $\Sol_\RR(\MM)$. This occurs if and only if the quasifree state constructed below is pure.
Under these circumstances, 
$\HH$ is first defined as a \emph{real} Hilbert space by completing $\Sol_\RR(\MM)$ in the norm $\|\phi\|_w=w(\phi,\phi)^{1/2}$ with inner product
$\ip{\phi_1}{\phi_2}_w=\Re w(\phi_1,\phi_2)$. 
Due to~\eqref{eq:saturation}, $\HH$ carries an isometry $J$ defined so that
\[
\Im w(\phi_1,\phi_2) = \ip{\phi_1}{J\phi_2}_w, \qquad \phi_i\in \HH,
\]
and which fulfils the conditions $J^2=-\II$, $J^\dagger =-J$. Hence $J$ is a complex structure on $\HH$, and we can make $\HH$ into a complex Hilbert space in which multiplication by $i$ is
(annoyingly) implemented by $-J$ and the inner product is
\[
\ip{\phi_1}{\phi_2}_\HH = \ip{\phi_1}{\phi_2}_w + i \ip{\phi_1}{J\phi_2}_w.
\] 
(See Appendix~\ref{appx:basic_fa} for more details.)
The map $K$ is just the natural inclusion of $\Sol_\RR(\MM)$ in $\HH$ and it is clear that $K\Sol_\RR(\MM)$ is dense in $\HH$ (hence $K\Sol_\RR(\MM)+iK\Sol_\RR(\MM)$ is also dense).
Verification of~\eqref{eq:1PS} is left as an exercise.

Returning to the general case, it may be shown that there is a state on $\Ac(\MM)$ given, as a formal series in $f$, by 
\begin{equation}\label{eq:QF}
\omega (e^{i\Phi(f)}) = e^{-W(f,f)/2},\qquad f\in\CoinX{\MM;\RR}.
\end{equation}
Expanding each side of~\eqref{eq:QF} in powers of $f$, and equating terms at each order, all expectation values of the form $\omega(\Phi(f)^n)$ are fixed. Arbitrary expectation values may then be formed using multilinear polarisation identities (see e.g.,~\cite{Thomas:2014} and references therein) and linearity.
It may be shown that all odd $n$-point functions vanish, while
\[
\omega(\Phi(f_1)\cdots\Phi(f_{2n}))= \sum_{G\in\mathcal{G}_{2n}} \prod_{e\in G}
W(f_{s(e)},f_{t(e)}),
\]
where $\mathcal{G}_{2n}$ is the set of directed graphs with vertices labelled $1,\ldots,2n$, such that each vertex is met by exactly one edge and the source and target of each edge obey $s(e)<t(e)$. An example for $n=4$ is given in Figure~\ref{fig:qf}. Another characterisation of the $n$-point functions is that all the \emph{truncated $n$-point functions} vanish for $n\neq 0,2$.
This type of state is described as \emph{quasifree}.

\begin{figure}[t]
	\begin{center}
		\begin{tikzpicture}
		\foreach \x in {1,...,8} {\node[draw,circle] (\x) at (\x,0) {$\x$};}
		\draw[->,>=stealth] (1) to[bend right] (4);
		\draw[->,>=stealth] (2) to[bend left] (7);
		\draw[->,>=stealth] (3) to[bend left] (5);
		\draw[->,>=stealth] (6) to[bend right] (8);
		\end{tikzpicture}
	\end{center}
	\caption{An example graph in $\mathcal{G}_8$.}\label{fig:qf}
\end{figure}
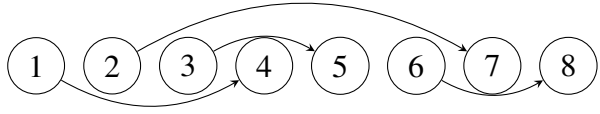 

\begin{exercise}
	Using~\eqref{eq:QF}, show that $\omega(\Phi(f)^{2n+1})=0$ for all $n\in\NN_0$, while 
	\[
	\omega(\Phi(f)^{2n}) = (2n-1)!!\, W(f,f)^n, \qquad\forall n\in\NN_0.
	\]
	Deduce that the sequence $\mu_n = \omega(\Phi(f)^n)$ satisfies the growth conditions in the Hamburger moment theorem. Therefore these are the moments of at most one probability measure -- which is of course a Gaussian probability distribution.
\end{exercise}

Although we have called $\omega$ a state, it is not yet clear that it has the required positivity property. This is most easily justified by noting that there is an explicit Hilbert space representation of $\Ac(\MM)$ containing a unit vector whose expectation values match those of $\omega$. To be specific, the Hilbert space is
the bosonic Fock space 
\begin{equation}\label{eq:qfF}
\FF(\HH) = \bigoplus_{n=0}^\infty \HH^{\odot n}
\end{equation}
over $\HH$, on which the field is represented according to the formula
\begin{equation}\label{eq:qfpi}
\pi_\omega(\Phi(f)) = a(KE f) + a^*(KEf)  ,\qquad f\in\CoinX{\MM;\RR},
\end{equation}
and
\[
\pi_\omega(\Phi(f)):=\pi_\omega(\Phi(\Re f)) + i\pi_\omega(\Phi(\Im f))
\]
for general complex-valued $f\in\CoinX{\MM}$.
Here $a(\phi)$ and $a^*(\psi)$ are the annihilation and creation operators on the Fock space which obey the CCRs
\begin{equation}\label{eq:CCRs}
[a(\phi),a^*(\psi)] = \ip{\phi}{\psi}_\HH\II, \qquad \phi,\psi\in\HH,
\end{equation}
on a suitable dense domain in $\FF(\HH)$ 
(note that $a$ is antilinear in its argument, and that $a(\phi)=a^*(\phi)^*$). Readers unfamiliar with the basis-independent notation used here should refer to Appendix~\ref{sec:Fock}. 

We will not give a detailed proof of the claims~(\ref{eq:QF}--\ref{eq:CCRs}), which would require consideration of operator domains.
At the level of formal calculation, however, it is easily checked that these operators do indeed lead to a representation of $\Ac(\MM)$.
\begin{exercise} Verify formally that $f\mapsto \pi_\omega(\Phi(f))$ is $\CC$-linear, obeys the field equation in the sense $\pi_\omega(\Phi(Pf))=0$, and is hermitian in the sense that $\pi_\omega(\Phi(f))^*=\pi_\omega(\Phi(\overline{f}))$.
\end{exercise}
For the CCRs, we compute
\[
[\pi_\omega(\Phi(f)),\pi_\omega(\Phi(g))] = \left(\ip{KEf}{KEg}_\HH - \ip{KEg}{KEf}_\HH\right)\II = 
\left(W(f,g)-W(g,f)\right)\II = iE(f,g)\II,
\]
using~\eqref{eq:CCRs},~\eqref{eq:1PS} and~\eqref{eq:WE}. Finally, 
it may be verified that the Fock vacuum vector $\Omega_\omega$ satisfies
\[
\ip{\Omega_\omega}{\pi_\omega(\Phi(f_1))\cdots \pi_\omega(\Phi(f_n))\Omega_\omega} = \omega(\Phi(f_1)\cdots \Phi(f_n)).
\]
Consequently, $\omega$ is seen to be a vector state on $\Ac(\MM)$ 
and the quadruple $(\FF(\HH), \DD_\omega,\pi_\omega,\Omega_\omega)$ is its GNS representation, where
the dense domain $\DD_\omega$ consists of finite linear combinations of finite products of operators $a^*(KEf)$ acting on $\Omega_\omega$.

The `one-particle space' $\HH$ may be interpreted as follows.  By~\eqref{eq:qfpi}, we see that
\[
\pi_\omega(\Phi(f))\Omega_\omega = a^*(KEf)\Omega_\omega = (0,KEf,0,\ldots)\in\FF(\HH) ,
\]
which is an eigenstate of $N$ with unit eigenvalue. Elements of $\HH$ can be identified with (complex linear combinations of) vectors generated from $\Omega_\omega$ by a single application of the field. Due to assumption~\ref{it:SFfieldeq}, these vectors may be identified with certain complex-valued solutions to the field equation, which may be regarded as wavepackets of `positive frequency modes' relative to a decomposition induced by the choice of quasifree state. 

\paragraph{Examples} We describe two important examples. The Minkowski vacuum state\footnote{A general definition of what a vacuum state should be will be given in Definition~\ref{def:vacuum}.} $\omega_0$ is a quasifree state with two-point function 
\[
W(f,g) = \int\frac{d^3\kb}{(2\pi)^3} \frac{\hat{f}(-k) \hat{g}(k)}{2\omega}, \qquad \text{where}~k^\bullet =(\omega,\kb),
\qquad \hat{g}(k) = \int d^4x\, e^{ik\cdot x}g(x),
\]
and $\omega=\sqrt{|\boldsymbol{k}|^2+m^2}$.\footnote{There is a notational conflict with the symbol used to denote states but it will always be clear from context which is meant.} The corresponding
one-particle space is $\HH=L^2(H^+_{m},d\mu)$, where
$H^+_m$ is the hyperboloid $k\cdot k=m^2$, $k^0>0$ in $\RR^4$,
and the measure of $S\subset H^+_m$ is
\[
\mu(S) = \int \frac{d^3\kb}{(2\pi)^3} \frac{\chi_S(k)}{2\omega}\qquad \chi_S(k)=\begin{cases} 1 & k\in S\\ 0 & \text{otherwise.}\end{cases}
\]
The map $K:\Sol_\RR(\MM)\to\HH$ is 
\[
KEg = \hat{g}|_{H^+_m},
\]
which already has dense range -- as mentioned above, this signals that the vacuum state is pure. Consequently, the vacuum representation $\pi_0$ is given by
\[
\pi_0(\Phi(g))=a(\hat{g}|_{H^+_m})+a^*(\hat{g}|_{H^+_m})\,, \qquad g\in\CoinX{\MM;\RR},
\]
which may be written in more familiar notation using sharp-momentum annihilation and creation operators obeying~\eqref{eq:aCCRs} using
\[
a(\hat{g}|_{H^+_m})=\int\frac{d^3\boldsymbol{k}}{(2\pi)^3}\frac{1}{\sqrt{2\omega}} \overline{\hat{g}(k)} a(\boldsymbol{k}),\qquad a^*(\hat{g}|_{H^+_m})=\int\frac{d^3\boldsymbol{k}}{(2\pi)^3}\frac{1}{\sqrt{2\omega}} \hat{g}(k) a^*(\boldsymbol{k})\,.
\]
Recalling that $\overline{\hat{g}(k)}=\hat{g}(-k)$ for real-valued $g$, we retrieve the field with sharp position as the operator-valued distribution
\[
\pi_0(\Phi(x))=\int \frac{d^3\boldsymbol{k}}{(2\pi)^3}\frac{1}{\sqrt{2\omega}}\left(a(\boldsymbol{k})e^{-ik\cdot x}+a^*(\boldsymbol{k})e^{ik\cdot x}
\right)\,.
\]

Our second example is the thermal state of inverse temperature $\beta$, with two-point function
\[
W_\beta(f,g) = \int\frac{d^3\kb}{(2\pi)^3}\frac{1}{2\omega} \left(\frac{\hat{f}(-k) \hat{g}(k)}{1-e^{-\beta \omega}} +
\frac{\hat{f}(k) \hat{g}(-k)}{e^{\beta \omega}-1}\right).
\]
Here, $\HH_\beta=\HH\oplus\HH$ with $\HH$ as before, and (cf.~\cite{Kay:1985a})
\[
(K_\beta \phi)(k) = \frac{(K\phi)(k)}{\sqrt{1-e^{-\beta\omega}}} \oplus  \frac{\overline{(K\phi)(k)}}{\sqrt{e^{\beta\omega}-1}} 
\]
(remember that $K_\beta$ only has to be real-linear!). The range of $K_\beta E$ is not dense in $\HH_\beta$, but its complex span is, reflecting the fact that the thermal states are mixed.
\begin{exercise}
	Check that the analogue of~\eqref{eq:1PS} holds for $\HH_\beta$, $K_\beta$ and $W_\beta$.
\end{exercise}

A nice feature of the algebraic approach is that, while the representations corresponding to the vacuum and thermal states are unitarily inequivalent, they can be treated `democratically' as states on the algebra $\Ac(\MM)$. There are many other quasifree states; indeed one can start with any state and construct its `liberation', the quasifree state with the same two-point function.

All quasifree representations carry a representation $\tilde{\pi}_\omega$ of the Weyl algebra as well, 
so that
\[
\tilde{\pi}_\omega(\Ws([f])) = e^{i\pi_\omega(\Phi(f))}
\]
and also 
\[
\pi_\omega(\Phi(f)) = \left.\frac{1}{i}\frac{d}{d\lambda} \Ws([\lambda f])\right|_{\lambda=0}.
\]
Remember that these relationships cannot hold literally either in $\Ac(\MM)$ or $\Wc(\MM)$, but here we see that they do hold in (sufficiently regular) representations.

Summarising, the quasifree states provide a class of states for which explicit Hilbert space representations may be given with the familiar Fock space structure, and in which the Weyl operators and smeared field operators are related as just described.

\section{The spectrum condition and Reeh--Schlieder theorem} \label{sec:spectrum}

In this section we begin to draw general conclusions about the properties of QFT in the AQFT framework, proceeding from the basic axioms and additional requirements that will be introduced along the way. We shall emphasise features that distinguish QFT from quantum mechanics. The starting point is a more detailed discussion of the action of Poincar\'e transformations.
\subsection{The spectrum condition}
Assumption~\ref{it:Poinc} required that
the Poincar\'e group should act by automorphisms of $\Ac(\MM)$. An important question concerning Hilbert space representations of the theory is whether or not these automorphisms are \emph{unitarily implemented}, i.e., whether there are unitaries $U(\rho)$ on the representation Hilbert space such that
\[
\pi(\alpha(\rho)A) = U(\rho) \pi(A) U(\rho)^{-1}.
\]
In such cases, we say that the representation is \emph{Poincar\'e covariant}.
As we now show, the GNS representation of a Poincar\'e invariant state $\omega$ is always Poincar\'e covariant. Here, Poincar\'e invariance of $\omega$ means that 
\[
\omega(\alpha(\rho)A) = \omega(A), \qquad\forall A\in\Ac(\MM),~\rho\in\Pc_0
\] 
(written equivalently as $\alpha(\rho)^*\omega=\omega$ for all $\rho\in\Pc_0$, with the star denoting the dual map). Of course, the same question can be asked of any automorphism or automorphism group on a $*$-algebra.
\begin{theorem}
	Let $\alpha$ be an automorphism of a unital $*$-algebra $\Ac$. If a state $\omega$ on $\Ac$ is invariant under $\alpha$, i.e., $\alpha^*\omega=\omega$, then $\alpha$ is unitarily implemented in the GNS representation of $\omega$ by a unitary that leaves the GNS vector invariant. Any group of automorphisms leaving $\omega$ invariant is unitarily represented in $\HH_\omega$. 
\end{theorem}
\begin{proof}
	Observing that $\omega(\alpha(A)^*\alpha(A))=(\alpha^*\omega)(A^*A)=\omega(A^*A)$, we see that $\alpha$ maps the GNS ideal $\Ic_\omega$ to itself. Therefore the formula
	\[
	U[A] = [\alpha(A)]
	\]
	gives a well-defined map $U$ from the GNS domain $\DD_\omega$ to itself, which fixes the GNS vector $\Omega_\omega=[\II]$ and is obviously invertible (consider $\alpha^{-1}$). Now
	\[
	\ip{U[A]}{U[B]} = \omega(\alpha(A)^*\alpha(A)) = \omega(A^*A) = \ip{[A]}{[B]},
	\]	
	so $U$ is a densely defined invertible isometry, and therefore extends uniquely to a unitary on $\HH_\omega$. The calculation
	\[
	\pi_\omega(\alpha(A))[B] = [\alpha(A)B] = [\alpha(A\alpha^{-1}(B))] = U[A\alpha^{-1}(B))] = U\pi_\omega(A)[\alpha^{-1}B] = U\pi_\omega(A)U^{-1}[B]
	\]
	shows that $U$ implements $\alpha$. 
	
	If $\beta$ is another automorphism leaving $\omega$ invariant, let $V$ be its unitary implementation as above. Then 
	\[
	UV[A]= [\alpha(\beta(A))] = [(\alpha\circ\beta)(A)]
	\]
	shows that $UV$ implements $\alpha\circ\beta$. 
\end{proof}
Among other things, this result may be applied to states that are   translationally invariant. Thermal equilibrium states, for example, are spatially homogeoneous but not Poincar\'e invariant because they have a definite rest frame. Although typical states are not translationally invariant, many interesting states are in a suitable sense local deviations from such invariant states. Even so, not all invariant states are physically acceptable. One way of doing narrowing the field is to require the spectrum condition:
\begin{definition}
	Let $\omega$ be a translationally invariant state so that the unitary implementation of the translation group $U(x)$ is \emph{strongly continuous}, i.e., the map $x\mapsto U(x)\psi$ is continuous from $\RR^4$ to $\HH_\omega$ for each fixed $\psi\in\HH_\omega$, where $x=(x^0,\ldots,x^3)\in\RR^4$.
	By Stone's theorem, there are four commuting self-adjoint operators $P^\mu$ such that 
	(lowering the index using the metric)
	\[
	U(x) = e^{iP_\mu x^\mu}.
	\]
	To any (Borel) subset $\Delta\subset \RR^4$ we may assign a projection operator $E(\Delta)$
	corresponding to the binary test of whether the result of $4$-momentum measurement $P^\mu$ would be found to lie in $\Delta$. The assignment $\Delta\mapsto E(\Delta)$ is a \emph{projection-valued measure}, and in fact one can write 
	\[
	U(x) = \int e^{ip_\mu x^\mu}dE(p^\bullet).
	\]
	The state $\omega$ is said to satisfy the \emph{spectrum condition} if the support of $E$ lies in the \emph{closed forward cone} $\overline{V^+}=\{p^\bullet\in\RR^4: p^\mu p_\mu\ge 0,~p^0\ge 0\}$, i.e.,
	\[
	\supp E \subset \overline{V^+}.
	\]
	This is sometimes expressed by saying that the joint spectrum of the momentum operators $P^\mu$ lies in $\overline{V^+}$. 
\end{definition}	 
An important consequence of the spectrum condition is that the definition of $U(x)$ can be extended to complex vectors $x\in \RR^4+iV^+$, with \emph{analytic dependence on $x$}: to be precise, $U(x)$ is strongly continuous on $\RR^4+i\overline{V^+}$ and holomorphic on $\RR^4+i\,V^+$, where $V^+=\text{int} \overline{V^+}$.  

One can check that the usual vacuum and thermal states of the free field obey the spectrum condition. More generally, we will make the following definition:
\begin{definition}\label{def:vacuum}
	A \emph{vacuum state}\footnote{The term `vacuum' is used in various ways by various authors, differing, for example, on whether Poincar\'e invariance is required as well. Somewhat remarkably, there is an algebraic criterion on the state that implies translational invariance, the spectrum condition and (if the state is pure) that the there is no other translationally invariant vector state in its GNS representation. See~\cite{Araki}.} is a translationally invariant state obeying the spectrum condition, whose GNS vector is the unique translationally invariant vector (up to scalar multiples) in the GNS Hilbert space. The corresponding GNS representation is called the \emph{vacuum representation.}
\end{definition}

\subsection{The Reeh--Schlieder theorem} We come to some general results that show how different
QFT is from quantum mechanics. For simplicity, suppose that the theory is given in terms of $C^*$-algebras. Let $\Omega$ be the GNS vector of a state obeying the spectrum condition (we drop the $\omega$ subscripts). Suppose the theory obeys the following condition: 
\begin{enumerate}[label=\bf A\arabic{enumi},leftmargin=*,widest=4]\setcounter{enumi}{5}
	\item {\bf Weak additivity}\\
	For any causally convex open region $\Ocal$,  $\pi(\Ac(\MM))$ is contained in the weak closure\footnote{A sequence of operators converges in the weak topology, $\text{w-}\lim A_n = A$, if $\ip{\psi}{A_n\varphi}\to \ip{\psi}{A\varphi}$ for all vectors $\psi$ and $\varphi$.}  of $\Bc_\Ocal$, the $*$-algebra generated by the algebras $\pi(\Ac(\Ocal+x))$ as $x$ runs over $\RR^4$. 
\end{enumerate}
Weak additivity asserts that arbitrary observables can be built as limits of
algebraic combinations of translates of 
observables in any given region $\Ocal$ (as one would expect in a quantum field theory). In combination with
the spectrum condition it has a striking consequence. Our proof is based on that of~\cite{Araki}.

\begin{theorem}[Reeh--Schlieder] Let $\Ocal$ be any causally convex bounded open region. Then (a)
	vectors of the form $A\Omega$ for $A\in\pi(\Ac(\Ocal))$ are dense in $\HH$; (b) if $A\in\pi(\Ac(\Ocal))$ annihilates the vacuum, $A\Omega=0$, then $A=0$. 
\end{theorem}
Part (a) says that $\Omega$ is cyclic for every local algebra; part (b) that it is also \emph{separating}. The existence of a cyclic and separating vector on a (von Neumann) algebra is the starting point of \emph{Tomita--Takesaki theory} (see, e.g.,~\cite{BratRob:vol1}). 

These are quite remarkable statements: if a local element of $\Ac(\Ocal)$ corresponds to an operation that can be performed in $\Ocal$, it seems that we can produce any state of the theory up to arbitrarily small errors by operations in any small region anywhere. To give an extreme example: by making a local operation in a laboratory on earth one could in principle modify the state of the theory to one approximating a situation in which there is a starship behind the moon, if such things can be modelled by the theory (e.g., as a complicated state of the standard model). It is an expression of how deeply entangled states in QFT typically are.  

\begin{proof} (a) Suppose to the contrary that the set of vectors mentioned is not dense. Then it has a nontrivial orthogonal complement, so there is a nonzero vector $\Psi\in\HH$ such that 
\[
\ip{\Psi}{A\Omega} = 0,\qquad \forall A\in\pi(\Ac(\Ocal)).
\]
Now let $\Ocal_1$ be a slightly smaller region with $\overline{\Ocal_1}\subset \Ocal$. For any $n\in\NN$ and any $Q_1,\ldots,Q_n\in\pi(\Ac(\Ocal_1))$ we have $U(x)Q_i U(x)^{-1}\in\pi(\Ac(\Ocal))$ for sufficiently small $|x|$, whereupon
\begin{equation}\label{eq:RSa}
\ip{\Psi}{U(x_1)Q_1 U(x_2-x_1) Q_2 U(x_3-x_2)\cdots U(x_n-x_{n-1})Q_n\Omega} = 0
\end{equation}
for sufficiently small $x_1,\ldots,x_n$. By what was said above, the function
\[
F: (\zeta_1,\ldots,\zeta_n)\mapsto \ip{\Psi}{U(\zeta_1)Q_1 U(\zeta_2) Q_2\cdots U(\zeta_n)Q_n\Omega} 
\]
is a continuous function on $(\RR^4)^n$ extending to an holomorphic function in $(\RR^4+iV^+)^n\subset (\CC^4)^n$, whose boundary value on $(\RR^4)^n$ moreover vanishes in some neighbourhood of the origin. The `edge of the wedge' theorem \cite{StreaterWightman,Vladimirov} implies that $F$ vanishes identically in its domain of holomorphicity, which means that the boundary value also vanishes identically. This means that \eqref{eq:RSa} holds for all $x_i$, or put another way, that
\[
\ip{\Psi}{\Bc_\Ocal\Omega}=0
\]
where $\Bc_\Ocal$ is the algebra generated by the algebras $\pi(\Ac(\Ocal_1+x))$ for $x\in\RR^4$. 
By weak additivity, we now know that $\ip{\Psi}{\pi(\Ac(\MM))\Omega} = 0$. But $\Omega_\omega$ is cyclic, so $\Psi=0$. This proves the first assertion. For the second, suppose that $A\in\pi(\Ac(\Ocal_1))$ annihilates $\Omega$. Choose $\Ocal_2$ causally disjoint from $\Ocal_1$, and note that for each $\psi\in\HH$ and all $B\in\pi(\Ac(\Ocal_2))$ we have
\[
\ip{A^*\psi}{B\Omega} = \ip{\psi}{AB\Omega} = \ip{\psi}{BA\Omega} = 0
\]
using Einstein causality. By the first part of the theorem, $A^*\psi$ is orthogonal to a dense set and therefore vanishes; as $\psi\in\HH$ is arbitrary we have $A^*=0$ and hence $A=0$. 
\end{proof}

\begin{corollary} 
	All nontrivial sharp local binary tests (with possible outcomes `success' or `failure') have a nonzero success probability in the vacuum state. (All local detectors exhibit `dark counts'). 
\end{corollary}
\begin{proof}
	A binary test can be represented by a projector $P$, with `failure' corresponding to its kernel. Suppose $P\in\pi(\Ac(\Ocal))$ is a projector with vanishing vacuum expectation value, i.e., a zero success probability. Then 
	\[
	\|P\Omega\|^2=  \ip{P\Omega}{P\Omega} = \ip{\Omega}{P\Omega} = 0,
	\]
	so
	$P\Omega=0$ and hence $P=0$ by the Reeh--Schlieder theorem (b). 
\end{proof}

\begin{corollary}\label{cor:vaccor}
For every pair of local regions $\Ocal_1$ and $\Ocal_2$ there are vacuum correlations between $\Ac(\Ocal_1)$ and $\Ac(\Ocal_2)$ (assuming $\dim \HH\ge 1$).
\end{corollary}
\begin{proof}
	For suppose to the contrary that there are states $\omega_i$ on $\Ac(\Ocal_i)$ such that 
	\[
	\ip{\Omega}{\pi(A_1)\pi(A_2)\Omega} = \omega_1(A_1)\omega_2(A_2)\qquad A_i\in\Ac(\Ocal_i).
	\]
	By setting $A_1=\II$, and then repeating for $A_2$, this implies 
	\[
	\ip{\Omega}{\pi(A_1)\pi(A_2)\Omega} = 	\ip{\Omega}{\pi(A_1) \Omega}	\ip{\Omega}{ \pi(A_2)\Omega}\qquad A_i\in\Ac(\Ocal_i)
	\] 
	Fixing $A_2$ and letting $A_1$ vary in $\Ac(\Ocal_1)$, the Reeh--Schlieder theorem gives
	\[
	\pi(A_2)\Omega = \ip{\Omega}{ \pi(A_2)\Omega}\Omega,\qquad A_2\in\Ac(\Ocal_2)
	\]
	which contradicts the Reeh--Schlieder theorem (a) unless $\dim\HH=1$. 
\end{proof}
The correlations indicated by Corollary~\ref{cor:vaccor} become small at spacelike separation due to the
\emph{cluster property} (a general feature of vacuum states in AQFT), which implies 
\[
\ip{\Omega}{\pi(A_1)\pi(\alpha(x)A_2)\Omega} \to	\ip{\Omega}{\pi(A_1) \Omega}	\ip{\Omega}{ \pi(A_2)\Omega}
\]
as $x\to\infty$ in spacelike directions, with exponentially fast convergence if the theory has a mass gap, i.e., $\sigma(P\cdot P)\subset\{0\}\cup[M^2,\infty)$ for some $M>0$. See~\cite[\S 4.3-4.4]{Araki}. The Reeh--Schlieder theorem does not present a very practical method for constructing starships.

\section{Local von Neumann algebras and their universal type} \label{sec:universaltype}

So far, we have encountered AQFTs given in terms of $*$-algebras or $C^*$-algebras. The theory is considerably enriched when expressed in terms of
von Neumann algebras. 

Consider any net of local $C^*$-algebras $\Ocal\mapsto\Ac(\Ocal)$ obeying~\ref{it:loc}--\ref{it:dyn}.
Given any state $\omega$ on $\Ac(\MM)$, form its GNS representation. For any open bounded causally convex $\Ocal\subset\MM$ we can define an algebra
\[
\Mf_\omega(\Ocal) := \pi_\omega(\Ac(\Ocal))'',
\]
namely, the \emph{double commutant} of the represented local algebra. Recall that
the commutant is defined for any subalgebra $\Bc$ of the bounded operators $\BB(\HH)$ on Hilbert space $\HH$ by 
\[
\Bc' = \{A\in\BB(\HH): [A,B]=0~\forall B\in\Bc\}.
\]
A basic result asserts that the algebra $\Mf_\omega(\Ocal)$ is also the closure of $\pi_\omega(\Ac(\Ocal))$ in the weak topology of $\HH_\omega$ -- as such, and because it contains the unit operator and is stable under the adjoint, $\Mf_\omega(\Ocal)$ is a (concrete) \emph{von Neumann algebra}. 
\begin{exercise}
	Show that for any subalgebra $\Bc$ of $\Bc(\HH)$ one has $\Bc\subset \Bc''$ and $\Bc'''=\Bc'$.
	Therefore $\Mf_\omega(\Ocal)''=\Mf_\omega(\Ocal)$. 
\end{exercise}

There is a good rationale for this construction. The norm topology of the $C^*$-algebra (which coincides with the norm topology $\BB(\HH_\omega)$ in a faithful representation) is quite stringent: we have already seen that there are no nonconstant continuous curves in the Weyl algebra. For example, a Weyl generator differs from all its nontrivial translates by operators of norm $2$. The situation is different in the weak topology: if translations are implemented in a strongly continuous way then matrix elements of an observable change continuously as the observable is translated. This motivates the weak topology as a better measure of proximity than the norm topology. Note also that Einstein causality implies that the commutant $\pi_\omega(\Ac(\Ocal))'$ contains all local algebras $\pi_\omega(\Ac(\Ocal_1))$ 
where $\Ocal_1$ is causally disjoint from $\Ocal$. As $\pi_\omega(\Ac(\Ocal))'=\pi_\omega(\Ac(\Ocal))'''=\Mf_\omega(\Ocal)'$, we see that also $\Mf_\omega(\Ocal_1) \subset \Mf_\omega(\Ocal)'$. It is natural to interpret $\Mf_\omega(\Ocal)'$ in terms of (limits of) observables localised in the causal complement of $\Ocal$ and further natural to assume that
there are no nontrivial observables that can be localised both in $\Ocal$ and its causal complement. This motivates an assumption that
\[
\Mf_\omega(\Ocal)\cap \Mf_\omega(\Ocal)' = \CC\II,
\]
that is, that the local algebras are \emph{factors}, so called because an equivalent statement is that
\[
\Mf_\omega(\Ocal)\vee \Mf_\omega(\Ocal)' = \BB(\HH),
\]
i.e., $\Mf_\omega(\Ocal)$ and its commutant generate [in the sense of von Neumann algebras] the full algebra of bounded operators on $\HH$.

The connection to von Neumann algebras permits a vast body of technical work to be brought to bear on AQFT, and indeed some of it was spurred by developments in AQFT. One of the most striking examples is the 1985 result of Fredenhagen~\cite{Fredenhagen:1985} that for physically reasonably states $\omega$ of reasonable theories obeying the AQFT axioms, every $\Mf_\omega(\Ocal)$ is a \emph{type  III${}_1$ factor}. Shortly afterwards, Haagerup~\cite{Haagerup:1987} proved that there was a unique \emph{hyperfinite} type III${}_1$ factor. A further paper of Buchholz, D'Antoni and Fredenhagen~\cite{BucDAnFre:1987} pointed out the local algebras of QFT are hyperfinite, therefore fixing them uniquely up to isomorphism. Remarkably, the local algebras themselves are completely independent of the theory! Therefore the distinction between different theories lies purely in the `relative position' of the local algebras within the bounded operators on Hilbert space.

To appreciate these results and their consequences, we need to delve into the classification of von Neumann factors. Type I factors are easily defined and familiar.
\begin{definition}\label{def:typeI}
	A von Neumann factor $\Mf$ on Hilbert space $\HH$ is of \emph{type I} if there is a unitary $U:\HH\to\HH_1\otimes\HH_2$, for some Hilbert spaces $\HH_i$, so that $\Mf=U^*(\Bc(\HH_1)\otimes\II_{\HH_2}) U$. 
	The type is further classified according to the dimension of $\HH_1$.
\end{definition}
One could say that type I factors are the natural playground of quantum mechanics: in an obvious way, they are the observables of one party in a bipartite system.
The result mentioned above indicates that QFT prefers type III for its local algebras.
\begin{definition}
	A von Neumann factor $\Mf$ on an infinite-dimensional separable Hilbert space $\HH$ is of \emph{type III} if every nonzero projection $E\in\Mf$ may be written in the form 
	\[
	E = WW^*, \qquad \text{for some $W\in\Mf$ obeying $W^*W=\II_\HH$.}
	\]
	(It is then the case that every two projections $E$ and $F$ in $\Mf$ can be written as $E=WW^*$, $F=W^*W$ for some $W\in\Mf$.) 
\end{definition} 
Type III${}_1$ is a further subtype whose definition would take us too far afield, but it should already be clear that types I and III are quite different. The foregoing definitions are quite technical in nature. However the proof that the local algebras have type III is founded on physical principles, namely that (a) the theory should have a description in terms of quantum fields, (b) that the $n$-point functions of these fields have a well-behaved scaling limit at short distances and (c) that there are no local observables localised at a point, other than multiples of the unit. In more detail, the assumptions are:
\begin{enumerate}[label=(\alph*)]
	\item There is a dense domain $\DD\subset \HH_\omega$ and a linear map $\phi$
	from real-valued test functions to symmetric operators $f\mapsto\phi(f)$ defined on $\DD$ and obeying $\phi(f)\DD\subset\DD$, so that $\phi(f)$ has a closure affiliated\footnote{That is, in the polar decomposition $\overline{\phi(f)}=U|\overline{\phi(f)}|$ of the closure $\overline{\phi(f)}$ of $\phi(f)$, the operator $U$ and any bounded function of $|\overline{\phi(f)}|$ belong to $\Mf_\omega(\Ocal)$.} 
	to $\Mf_\omega(\Ocal)$ for any $\Ocal$ containing $\supp f$. Furthermore, the $n$-point functions of the fields $\phi(f)$ define distributions
	\[
	W_n(f_1,\ldots,f_n)= \ip{\Omega_\omega}{\phi(f_1)\cdots\phi(f_n)\Omega_\omega}.
	\]
	\item Defining scaling maps on test functions by $(\beta_{p,\lambda}f)(x)=f((x-p)/\lambda)$, there should be a positive monotone function $\nu$ so that the scaling limit $n$-point functions
	\[
	W_n^{\text{s.l.}}(f_1,\ldots,f_n) = \lim_{\lambda\to 0^+} \nu(\lambda)^n W_n(\beta_{p,\lambda}f_1,\ldots,\beta_{p,\lambda} f_n)
	\] 
	exist and
	satisfy the (vacuum) Wightman axioms~\cite{StreaterWightman}. One may think of this theory as living in the tangent space at $p$; it is not usually the same theory as the one we started with.
	\item It is required that $\bigcap_{\Ocal\owns p}\Mf_\omega(\Ocal)=\CC\II$, where the intersection is taken over all local regions containing the point $p$.
\end{enumerate}
Assuming condition (a), we say that $\omega$ has a \emph{regular scaling limit} at $p\in\MM$ if the conditions (b) and (c) hold. Fredenhagen's result~\cite{Fredenhagen:1985} can now be stated precisely.
\begin{theorem}\label{thm:universaltype}
	Let $\Ocal$ be a double-cone. If the state $\omega$ has a regular scaling limit at some point of the spacelike boundary of $\Ocal$ (i.e., its equatorial sphere) then $\Mf_\omega(\Ocal)$ has type III${}_1$.   
\end{theorem} 
Conversely, if any local von Neumann algebra of a double-cone has type other than III${}_1$, then at least one of the assumptions (b),(c) must fail at every point on its spacelike boundary, or (a) fails. One application of this argument has been to prove that `SJ states' of the free field on double-cone regions do not extend to Hadamard states on a larger region - typically because the stress-energy tensor diverges as the boundary is approached~\cite{FewsterVerch:2013}. We will say a little more about this below.

Various key distinctions between quantum mechanics and QFT can be attributed to the differences between factors of types I and III, and the fact that the local algebras of QFT are typically type III. To conclude this section, we collect some properties of type III factors and their consequences for QFT, gathered under some catchy slogans. We assume that the Hilbert space on which the type III factors act is infinite-dimensional and separable. 

\paragraph{All eigenvalues have infinite degeneracy}
No type III factor can contain a finite-rank projection; for if $E$ 
is finite rank, then writing $E=WW^*$ for an isometry $W$, we see that 
$W^*EW=W^*WW^*W=\II_\HH$ is also finite rank, contradicting our assumption on the dimension of $\HH$.
It follows that no self-adjoint element of a local algebra $\Mf(\Ocal)$ can have a finite-dimensional eigenspace, which is far removed from the situation of elementary textbook quantum mechanics. In particular, if $\Omega$ is the vacuum vector,
then the projection $\ket{\Omega}\bra{\Omega}$ belongs to no local algebra $\Mf(\Ocal)$, nor to any commutant $\Mf_\omega(\Ocal)'$ [because the commutant of a type III factor is also of type III].

\paragraph{Local states are impure} Any state $\omega$ on $\Ac(\MM)$ restricts in a natural way to a state $\omega|_{\Ac(\Ocal)}$ on each local algebra $\Ac(\Ocal)$. 
What sort of state is it? To start, we note the simple:
\begin{lemma}
	Suppose $\omega$ is a pure state on a $C^*$-algebra $\Ac$. Then the von Neumann algebra $\Mf=\pi_\omega(\Ac)''$ in the GNS representation of $\omega$ is a type I factor. 
\end{lemma}
\begin{proof}
The GNS representation of a pure state is irreducible. So the commutant $\pi_\omega(\Ac)'$ consists only of multiples of the unit, and $\Mf=\BB(\HH)$,
which is evidently type I.
\end{proof}
In the light of this result, it may not be so much of a surprise that:
\begin{theorem}
	Under the hypotheses of Theorem~\ref{thm:universaltype}, the algebra
	$\pi_{\omega|_{\Ac(\Ocal)}}(\Ac(\Ocal))''$ is of type III. Therefore $\omega|_{\Ac(\Ocal)}$ is not pure, nor is it a normal state in the GNS representation of a pure state of $\Ac(\Ocal)$. 
	Furthermore, in the case where $\Mf_\omega(\Ocal)$ is a factor, the previous statements
	hold if $\omega$ is replaced by any state in its folium. 
\end{theorem}
\begin{proof} This requires a few standard results from von Neumann theory and can be found as Corollary 3.3 in~\cite{FewsterVerch:2013}.
\end{proof}
The SJ states mentioned above induce pure states and therefore it follows that they cannot be induced as restrictions of states on $\Ac(\MM)$ with good scaling limit properties. See~\cite{FewsterVerch:2013}. 

\paragraph{Local experiments can be prepared locally} 
Suppose $E$ is any projection in $\Mf(\Ocal)$ and let $\omega$ be any state
on the $C^*$-algebra formed as the closure of $\bigcup_{\Ocal}\Mf(\Ocal)$ in $\BB(\HH)$. By the type III property, we may write
\[
E= WW^*
\]
for some isometry $W\in\Mf(\Ocal)$. Then the modified state \emph{(Exercise: check that it is a state!)}
\[
\omega_W(A):= \omega(W^* A W)
\]
obeys
\[
\omega_W(E) = \omega(W^*WW^*W) =1
\]
so the yes/no test represented by $E$ is certainly passed in state $\omega_W$.
On the other hand, if $A\in\Mf(\Ocal_1)$ where $\Ocal_1$ is causally disjoint from $\Ocal$, then 
$A\in\Mf(\Ocal)'$ and so
\[
\omega_W(A) = \omega(W^*AW) = \omega(W^*WA) = \omega(A).
\]
So by changing the state in this way we can prepare a state with a desired property in our lab without changing the rest of the world.

For a brief survey of this and other features of type III factors and their relevance to QFT, see~\cite{Yngvason:2005,Yngvason:2015}.

\section{The split property}\label{sec:split}

The fundamental postulates of relativistic physics entail that two spacelike separated laboratories should not be able to communicate. Consequently,
an experimenter situated in one of these laboratories should be able to conduct experiments independently of the actions (or even the presence) of an experimenter in the other region. The Einstein causality condition~\ref{it:Eins} reflects this idea, because any observable
from one local algebra will commute, and be simultaneously measurable, with any observable from the other.  However, this is far from being the only way in which the two regions must be independent. For example, the two observers should be able to prepare their experiments for measurement independently, too.
In quantum mechanics this independence is modelled by assigning each of the two local systems their own Hilbert space, $\HH_1$ and $\HH_2$, on which the observables in the two laboratories respectively act. If the experimenters prepare states $\psi_i$, the global state is then taken to be $\psi_1\otimes\psi_2$ on the tensor product $\HH_1\otimes\HH_2$. 
This section describes the \emph{split property} which provides conditions under which a similar level of independence may be established in QFT.

First, we introduce some terminology: two von Neumann algebras $\Mf_1$ and $\Mf_2$ acting on a Hilbert space $\HH$ are said to form a \emph{split inclusion} if there is a type I von Neumann factor $\Nf$ so that
\[
\Mf_1\subset \Nf\subset \Mf_2.
\]
By Definition~\ref{def:typeI}, this means that there is a unitary $U:\HH\to \HH_1\otimes\HH_2$ with $\Nf = U^* (\BB(\HH_1)\otimes\II_{\HH_2})U$ for some Hilbert spaces $\HH_1$ and $\HH_2$. The commutant of $\Nf$ is easily described:
\[
\Nf' = U^* (\II_{\HH_1}\otimes\BB(\HH_2))U
\]
and it follows from the split inclusion that
\[
\Mf_1\subset  U^* (\BB(\HH_1)\otimes\II_{\HH_2})U,\qquad   \Mf_2' \subset U^* (\II_{\HH_1}\otimes\BB(\HH_2))U.
\]
Suppose that states $\omega_1$ and $\omega_2$ are given on $\Mf_1$ and $\Mf_2'$ that may be expressed in terms of density matrices $\rho_i$ on $\HH_i$, so that
\[
\omega_1(A)= \tr ((\rho_1\otimes\II_{\HH_2})UAU^*), \qquad \omega_2(B)= \tr ((\II_{\HH_1}\otimes\rho_2)UBU^*)
\]
for $A\in\Mf_1$, $B\in\Mf_2'$. Then there is an obvious joint state
\[
\omega(C) = \tr ((\rho_1\otimes\rho_2)UAU^*)
\]
with the property \emph{(Exercise: prove it!)} that
\[
\omega(AB)= \omega_1(A)\omega_2(B) = \omega(BA), \qquad \text{for all}~A\in\Mf_1, ~B\in\Mf_2'.
\]
With a little more technical work it can be shown that this construction is possible whenever $\omega_i$ are given as density matrices on $\HH$ (see e.g.,~\cite{Fewster_Abh:2016} for an exposition). 

Returning to QFT, we make the following definition.
\begin{definition}
	A net $\Ocal\mapsto\Mf(\Ocal)$ of von Neumann algebras has the \emph{split property} if, whenever $\overline{\Ocal_1}\subset \Ocal_2$, the inclusion $\Mf(\Ocal_1)\subset \Mf(\Ocal_2)$ is a split inclusion. 
\end{definition}
Note that we only require a split inclusion when $\Ocal_2$ contains the closure of $\Ocal_1$, so that there is a `safety margin' or collar around $\Ocal_1$ within $\Ocal_2$. We can also make the less restrictive assumption that the split inclusion holds when this safety margin has a minimum size, in which case one speaks of the \emph{distal split property}. 

If the split property holds, and $\Ocal_1$ and $ \Ocal_3$ are spacelike separated pre-compact regions whose closures do not intersect, then we may certainly find a pre-compact neighbourhood $\Ocal_2$ of $\overline{\Ocal}_1$ within the causal complement of $\overline{\Ocal_3}$. Then $\Mf(\Ocal_1)\subset \Mf(\Ocal_2)$ is split and $\Mf( \Ocal_3)\subset \Mf(\Ocal_2)'$. It follows that experimenters in $\Ocal_1$ and $ \Ocal_3$ are able to independently prepare and measure observables of the field theory. Moreover, there is an isomorphism of von Neumann algebras
\[
\Mf(\Ocal_1)\bar{\otimes}\Mf( \Ocal_3)\cong \Mf(\Ocal_1)\vee \Mf( \Ocal_3)
\]
extending the map $A\otimes B\mapsto AB$, 
where the left-hand side is a spatial tensor product of von Neumann algebras and the right-hand side is the von Neumann algebra generated by sums and products of elements in $\Mf(\Ocal_1)$ and $\Mf(\Ocal_3)$.

In situations where the Reeh--Schlieder theorem applies, and there is a vector that is cyclic and separating for all causally convex bounded regions, the split inclusions have more structure and are called \emph{standard split inclusions}. A deep analysis by Doplicher and Longo~\cite{DopLon:1984} shows, among many other things, that there is a canonical choice for the type I factor appearing in standard split inclusions. 

The split property is enjoyed by free scalar fields of mass $m\ge 0$ and the observable algebra for the Dirac field, but also for certain interacting models in $1+1$-dimensions (see~\cite{Lech_chap:2015} for a survey and exposition).
It is intimately connected to the way in which the number of local degrees of freedom available to the theory grows with the energy scale.
These are expressed technically in terms of various \emph{nuclearity conditions}, which we will not describe in detail here (but see~\cite{Fewster_Abh:2016} for discussion and a relation to yet a further topic -- Quantum Energy Inequalities). Instead, we limit ourselves to some examples involving a theory comprising countably many independent free scalar fields of masses $m_r$ ($r\in \NN$). It may be shown, for instance, that this theory has the split property if the function
\[
G(\beta):= \sum_{r=1}^\infty e^{-\beta m_r/4}
\]
is finite for all $\beta>0$ and grows at most polynomially in $\beta^{-1}$ as $\beta\to 0^+$. This is the case if $m_r=rm_1$, for instance. On the other hand, if 
\[
m_r = (2d_0)^{-1}\log(r+1)
\]
for some constant $d_0>0$ then the series defining $G(\beta)$ diverges for $\beta\le 8d_0$. 
Further analysis~\cite[Thm 4.3]{DAnDopFreLon:1987} shows that the split property fails in this situation, but that the distal split condition holds provided that the `safety margin' is sufficiently large. For concentric arrangements of double cones of base radii $r$ and $r+d$, splitting fails if $d<d_0$ and succeeds if $d>2d_0$. 

The overall message is that a (well-behaved) tensor product structure across regions at spacelike separation can only be expected in well-behaved QFTs and with a safety margin between the regions. It turns out that the split property is closely related both to the existence of well-behaved thermal states (absence of a Hagedorn temperature) and to whether the theory satisfies quantum energy inequalities -- see~\cite{Fewster_Abh:2016} for discussion and original references. Intuitively, the reason for this is that there is a cost associated with disentangling the degrees of freedom between the two regions. For example, a joint state `glued together' from two states given on their local algebras might be expected to have a higher energy than the sum of the two original energies, with the excess in energy being higher if the `gluing' has to take place over a smaller gap between the regions. If the number of degrees of freedom available grows excessively with the energy scale, it becomes impossible to achieve the gluing with a finite excess energy, or possible only if the regions are sufficiently distant from one another.

\section{Superselection sectors}\label{sec:sectors}
We come back to the discussion of superselection sectors, mentioned in section~\ref{sec:PosQM}. In the algebraic viewpoint on quantum theory, superselection sectors correspond to a class of unitarily inequivalent representations of the algebra of observables. One of the major structural results in algebraic quantum field theory \cite{DHR1,DHR2,DHR3,DHR4,DR89} was to show that these sectors are related to irreducible representations of some compact Lie group $G$ (the global gauge group). The key idea of algebraic QFT is that \textit{all the relevant information is contained in the net of observables}, and from this net one can construct an algebra of fields $\F(\Ocal)$, which then contains non-observable objects, e.g. smeared Dirac fields $\psi(f)$ (see section~\ref{sec:examples}). This algebra is \textit{uniquely fixed by the net} and it carries the action of the gauge group $G$, so that local algebra $\Ac(\Ocal)$ consists of elements of $\Fc(\Ocal)$ invariant under $G$. The reconstruction of both $\Fc$ and $G$ is achieved through the DHR (Doplicher--Haag--Roberts) analysis \cite{DHR1,DHR2,DHR3,DHR4} together with the Doplicher--Roberts reconstruction theorem \cite{DR89}. The brief exposition in Sections~\ref{interest}-\ref{sec:intertwiners} is based on~\cite{KlausSuperselection}.

As motivating example, we consider the \emph{complex scalar field}, whose algebra
$\Cc(\MM)$ was described in Sec.~\ref{sec:examples}. Its vacuum representation $\pi$ is given by
\begin{align*}
\pi(\Phi(g)) &= a(\hat{g}|_{H_m^+})\otimes \II + \II\otimes a^*(\hat{g}|_{H_m^+})\\
\pi(\Phi^\star(g)) &= a^*(\hat{g}|_{H_m^+}) \otimes \II + \II\otimes a(\hat{g}|_{H_m^+})
\end{align*}
for real-valued $g$ [with $\Phi(g)=\Phi(\Re g) + i \Phi(\Im g)$ in general, compare with the real scalar field in Section~\ref{sec:quasifree}], on $\FF(\Hcal)\otimes \FF(\Hcal)$, where $\Hcal=L^2(H^+_{m},d\mu)$ and $a^*,a$ are creation and annihilation operators on the Fock space  $\FF(\Hcal)$. The charge operator is 
\[
Q = N\otimes \II - \II \otimes N\,,
\]
where $N$ is the number operator on $\FF(\Hcal)$ (see Appendix~\ref{sec:Fock} for definition). The theory has a global $U(1)$ gauge symmetry generated by $Q$:
\begin{equation}
\label{eq:QPhi} 
e^{i\alpha Q}\pi(\Phi(f)) e^{-i\alpha Q} = e^{-i\alpha} \pi(\Phi(f)),\qquad \alpha\in\RR,
\end{equation} 
which implements the automorphisms $\eta_\alpha$ described earlier, 
and the Fock space decomposes into charged sectors
\[
\FF(\Hcal)\otimes \fF(\Hcal) = \bigoplus_{q\in\ZZ} \HH_q,
\]
each $\HH_q$ being the eigenspace of $Q$ with eigenvalue $q$. Assuming that $Q$ is conserved in all interactions available to the observer, the argument
described in Sec.~\ref{sec:PosQM} shows that the observables of the theory should commute with $Q$ and be block-diagonal in this decomposition. Clearly, this is not the case for the smeared fields $\pi(\Phi(f))$ or $\pi(\Phi^\star(f))$, which are consequently unobservable. On the other hand, operators of the form $\pi(\Phi^\star(f)\Phi(f))$ is gauge-invariant, as is (any smearing of) the Wick product of the field $\pi(\Phi^\star(x))$ with $\pi(\Phi(x))$. 

The main point of interest for us is that the local observable algebras $\Cc_{\text{obs}}(\Ocal)$ are gauge-invariant by definition and therefore $\pi(\Cc_{\text{obs}}(\Ocal))$ consists of block-diagonal operators. More than that, we can see that there are representations $\pi_q$ of $\Cc_{\text{obs}}(\Ocal)$ on each $\HH_q$, given by 
	\[
	\pi_q(A) = P_q \pi(A) P_q,
	\]
	where $P_q$ is the orthogonal projector onto $\HH_q$ within $\FF(\Hcal)\otimes \FF(\Hcal)$. In particular, $\pi_0$ is a representation in $\HH_0$, the charge-zero sector, which  contains the vacuum vector $\Omega\otimes\Omega\in\FF(\Hcal)\otimes\FF(\Hcal)$ and will be called the vacuum sector for the observables.

We now have a whole family of representations of the algebras $\Cc_{\text{obs}}(\Ocal)$ on different Hilbert spaces. It is easy to argue (at least heuristically) that these are mutually unitarily inequivalent, because the charge operator itself can be regarded as a limit of local observables, i.e., local integrals of the Noether charge density associated with the global gauge symmetry. Therefore, if there were a unitary $U:\HH_q\to\HH_{q'}$ obeying $\pi_q'(A) = U\pi_q(A)U^{-1}$ for all $A\in\Cc_{\text{obs}}(\Ocal)$ and all local regions $\Ocal$, it would follow that
the charge operators in the two representations should be equivalent under the same mapping, giving
\[
q'\II_{\HH_{q'}} = q U\II_{\HH_{q}} U^{-1} = q \II_{\HH_{q'}}
\] 
which can only happen in the case $q'=q$. The physical distinction between different sectors is precisely that they have different charge content, and the central insight of the DHR programme is that the relevant charges might be gathered in some local region, outside which the distinction can be, as it were, gauged away.  An extra charged particle behind the moon ought not to change our description of particle physics on earth.\footnote{This picturesque statement applies to confined charges rather than those with  long range interactions such as electric charges. We discuss this issue briefly later on.} As we will describe this physical insight allows a remarkable reversal of the process we have just followed: instead of starting with an algebra of unobservable fields transforming under a global gauge group and obtaining from it a collection of unitarily inequivalent representations of the algebra of observables, one starts from a suitable class of representations of the observable algebras and attempts to reconstruct the unobservable fields and the a unifying gauge group. 

\subsection{Representations of interest in particle physics}\label{interest}
In the first step of DHR analysis we want to single out a class of representations relevant for study of superselection sectors. We focus on theories without long range effects, for example strong interactions in hadron physics. Loosely speaking, the representations of interest correspond to states that are generated from the vacuum by (possibly unobservable) local field operators. The intuition is that these states have different global charges (e.g. QCD charges) but which are localized in compact regions.

Let $\Ocal\mapsto \Ac(\Ocal)$ be a net of $C^*$-algebras
and let $\Ac$ be the quasilocal algebra (to reduce clutter, we write $\Ac(\M)\equiv \Ac$ in this section). Let $\pi_0$ be the vacuum representation (assumed here to be faithful).

We are interested in reconstructing the field algebras $\F(\Ocal)$ realized as von Neumann algebras of operators on some common Hilbert space $\Hcal_{tot}$, i.e. $\F(\Ocal)\subset \Bcal(\Hcal_{tot})$. We want them to satisfy (among others) the following properties:
	\begin{enumerate}[label=\bf F\arabic{enumi},leftmargin=*,widest=4] 
	\item $\F(\Ocal)$ carries a strongly continuous representation of the (covering group) of the Poincar\'e group and there exists a unique vector $\Omega\in \Hcal_{tot}$ (up to  phase) that is invariant under this action.
	\item There exists a compact group $G$ (the \textit{gauge group}) and a strongly continuous faithful representation $U$ of $G$ in  $\Hcal_{tot}$ such that the automorphism $A\mapsto \alpha_g(A)= \mathrm{Ad}_{U(g)}(A)=U(g)A U(g)^{-1}$ obeys
\[
	\alpha_g(\F(\Ocal))=\F(\Ocal)\,,\qquad U(g)\Omega=\Omega\,,
\]
	and the $U(g)$'s commute with the representation of the Poincar\'e group mentioned above.
	
	On general grounds, the Hilbert space $\Hcal_{tot}$ decomposes into \textit{superselection sectors} as:
	\begin{equation}\label{eq:H_decomp}
	\Hcal_{tot}=\bigoplus_\sigma \tilde{\Hcal}_\sigma\otimes \Hcal_\sigma\,,\qquad
	U(g) = \bigoplus_\sigma  U_\sigma(g)\otimes \II_{{\Hcal}_\sigma}\,,
	\end{equation}
	where the sum is taken over equivalence classes $\sigma$ of unitary irreps $(\tilde{\Hcal}_\sigma,U_\sigma)$ of $G$, with representation space $\tilde{\Hcal}_\sigma$ and $\Hcal_\sigma$ reflects the multiplicity with which these representations appear (including the possibility that  $\Hcal_\sigma$ has zero dimension, in which case $\sigma$ does not appear).
	
	\item Reconstructing the observables: $\Hcal_{tot}$ carries a representation $\pi$ of $\Ac(\Ocal)$ and
	\[
	\pi(\Ac(\Ocal))''=\Fc(\Ocal)\cap U(G)' =\{A\in\F(\Ocal), \alpha_g(A)=A, \forall g\in G\}.
	\]
	By construction, this representation decomposes w.r.t.\,~\eqref{eq:H_decomp} as
	\[
	\pi = \bigoplus_\sigma  \II_{\tilde{\Hcal}_\sigma}\otimes \pi_\sigma
	\]
	where each $\pi_\sigma$ is a representation of $\Ac(\Ocal)$ on $\Hcal_\sigma$.
	In particular, the representation corresponding to the trivial representation of $G$ should coincide with $\pi_0$, the vacuum representation.  
\end{enumerate}

The task of reconstructing the field algebra therefore amounts to determining the relevant representations $\pi_\sigma$, and corresponding (irreducible) representations $U_\sigma$ of $G$, which
can then be assembled to form $\Hcal_{tot}$. It may be shown that the representations $\pi_\sigma$ of interest are those which satisfy the following criterion (due to DHR \cite{DHR3}):
\vskip 0.1in
\begin{df}
	A representation $\pi$ of a net of $C^*$-algebras $\Ocal\mapsto\Ac(\Ocal)$ is a \emph{DHR representation} if it is Poincar\'e covariant and the following holds:
	\be\label{DHRcrit}
	\pi\big|_{\Ac(\Ocal')}\cong \pi_0\big|_{\Ac(\Ocal')}
	\ee
	for all double-cones\footnote{The criterion can be weakened so as to refer to sufficiently large double-cones.} $\Ocal$, where $\Ocal'$ is the causal complement of $\Ocal$ and $\cong$ means unitary equivalence, i.e. there exists a unitary operator between the appropriate Hilbert spaces $V:\Hcal_{0}\rightarrow \Hcal_{\pi}$ such that $V\pi_0(A)=\pi(A) V$ for all $A\in \Ac(\Ocal')$.
\end{df}
In the above definition, the region $\Ocal'$ is unbounded, so $\Ac(\Ocal')$ is not given a priori in the specification of the theory. It is defined to be the $C^*$-algebra generated by all local algebras $\Ac(\Ocal_1)$ for bounded $\Ocal_1\subset \Ocal'$. 

The intertwining unitary $V$ has the interpretation of a \emph{charge-carrying field}. To get some intuition about these objects, consider the example of a complex scalar field mentioned at the beginning of this section.
Consider a test function $f\in\CoinX{\MM;\RR}$ with $\Phi^\star(f)\neq 0$. Firstly, we note that the charge is related to the phase of the field $\pi(\Phi^\star(f))$, so in order to extract an intertwiner $V$, we consider the polar decomposition $\overline{\pi(\Phi^\star(f))}=V_f |\overline{\pi(\Phi^\star(f))}|$, where the overline denotes an operator closure. The partial isometry $V_f$ may be taken to be unitary because $\overline{\pi(\Phi^\star(f))}$ is normal; acting on vectors in the vacuum sector $\HH_0$, it has the property $V_f^q \pi_0(A)= \pi_q(A)V_f^q$, where $A\in\Cc(\M)$ and $\pi_q$ is the representation with charge $q\in\ZZ$. We see that $V_f$ creates a single unit of charge and the support of $f$ determines the region, where the charge is localized. If $f$ is supported in a double cone $\Ocal$ and we take $f_1$ supported inside another double cone $\Ocal_1$, then $V_{f_1}^q V_f^{-q}$ transports $q$ units of charge from $\Ocal$ to $\Ocal_1$. 

There exists a generalization of the DHR framework to the situation, where charges are not localized in bounded regions, but rather in cone-like regions (\emph{space-like} cones formed as a causal completions of spacial cones). This is relevant if one wants to apply this analysis to theories with long-range interactions, e.g. quantum electrodynamics (QED). It is expected that the electron is a charged particle with this type of localization. The corresponding version of the DHR construction has been developed by Buchholz and Fredenhagen (BF analysis) in \cite{BF82} and the full analysis of superselection sectors for QED has been achieved in \cite{BR14}.

We may now turn things around and phrase  our problem as follows: starting with the abstract algebra of observables and its vacuum representation we want to classify the equivalence classes of its (irreducible) representations satisfying \eqref{DHRcrit}. Following the literature \cite{Haag,DHR1,DHR2,DHR3,DHR4}, we call each of these equivalence classes a superselection sector (also called \textit{charge superselection sectors}, whereupon the labels $\sigma$ are referred to as \textit{charges}, though they need not be numbers). 
%
%

 A special case is the situation where $(U_\sigma,\tilde{\Hcal}_\sigma)$ is a one-dimensional representation. 
  Remarkably, such \emph{simple sectors} are distinguished by the following property of $\pi_\sigma$:
\begin{df}
	We say that a representation $\pi$ satisfies \emph{Haag duality}, if  
				\be\label{Haagdual}
				\pi(\Ac(\Ocal'))''=\pi(\Ac(\Ocal))'\cap \pi(\Ac)''\,,
				\ee 
				for any double-cone $\Ocal$. 
				If $\pi$ is irreducible then the intersection with $\pi(\Ac)''$ is superfluous, 
				whereupon one also has
				\be\label{Haagdual2}
			 \pi(\Ac(\Ocal'))'=\pi(\Ac(\Ocal))''\,.
			 \ee
			(\emph{Exercise!})
\end{df}
One can show that $\pi_\sigma$ satisfies Haag duality if and only if $\tilde{\Hcal}_\sigma$ is one-dimensional \cite{DHR1}.  In particular, the vacuum sector is always simple; however, if the 
global gauge group is nonabelian, there will necessarily be some non-simple sectors.

\begin{exercise}
	If $\Ac$ is a $C^*$-algebra with Hilbert space representations $\pi_1$ and $\pi_2$ on Hilbert spaces $\HH_1$ and $\HH_2$, define the representation $(\pi_1\oplus\pi_2)(A)=\pi_1(A)\oplus \pi_2(A)$ on $\HH_1\oplus\HH_2$. Compute
	the commutant $(\pi_1\oplus\pi_2)(\Ac)'$ and double commutant $(\pi_1\oplus\pi_2)(\Ac)''$ within $\BB(\HH_1\oplus\HH_2)$, which may be regarded as consisting of $2\times 2$ `block matrices' of operators, and compare the results with $\pi_1(\Ac)''\oplus \pi_2(\Ac)''$. 
\end{exercise}

\subsection{Localized endomorphisms}
We now want to introduce some algebraic structures on the space  of representations of interest, following closely the exposition presented in \cite{KlausSuperselection}. Our standing assumptions are that the vacuum representation $\pi_0$ is faithful and irreducible, satisfies Haag duality, and that the local algebras $\pi_0(\Ac(\Ocal))$ are weakly closed, $\pi_0(\Ac(\Ocal))''=\pi_0(\Ac(\Ocal))$, so they are actually von Neumann algebras.\footnote{If weak closure does not hold, one can apply this discussion to the net $O\mapsto \Mf_0(\Ocal):=\pi_0(\Ac(\Ocal))''$.} Together with Haag duality, this gives
\be\label{Haagdual3}
			 \pi_0(\Ac(\Ocal'))'=\pi_0(\Ac(\Ocal))
\ee
for every double-cone $\Ocal$.

Let us fix a double cone $\Ocal$. Consider a representation $\pi$ of $\Ac$ satisfying the DHR criterion with a unitary $V:\Hcal_0\rightarrow\Hcal_{\pi}$ implementing the equivalence in~\eqref{DHRcrit} for $\Ocal$, and
define a representation $\tilde{\pi}$ on $\Hcal_0$ by
\[
\tilde{\pi}(A)=V^{-1}\pi(A) V\,,\quad A\in\Ac\,.
\]
Take $\Ocal_1\supset \Ocal$ and $A\in\Ac(\Ocal_1)$, $B\in\Ac(\Ocal'_1)$. Since $\Ac(\Ocal'_1)\subset \Ac(\Ocal')$ and $\tilde{\pi}=\pi_0$ on $\Ac(\Ocal')$ (by the DHR criterion), we have
\[
[\pi_0(B),\tilde{\pi}(A)]=\tilde{\pi}([B,A])=0\,,
\]
so $\tilde{\pi}(A)\in \pi_0(\Ac(\Ocal'_1))'=\pi_0(\Ac(\Ocal_1))$, where the last assertion follows from Haag duality~\eqref{Haagdual3}. We conclude that
 $\tilde{\pi}(\Ac(\Ocal_1))\subset \pi_0(\Ac(\Ocal_1))$, so $\tilde{\pi}(\Ac)\subset \pi_0(\Ac)$, as local operators are dense in $\Ac$. Since $\pi_0$ is faithful, \emph{there exists an endomorphism $\rho$ of the abstract $C^*$-algebra $\Ac$ with $\rho=\pi_0^{-1}\circ \tilde{\pi}$}. 

There is an obvious equivalence relation on endomorphisms:
\begin{equation}\label{equiv}
\rho_1\sim\rho_2 \Leftrightarrow \rho_1=\iota\circ\rho_2\,, 
\end{equation}
for some  inner automorphism $\iota$, i.e. one that can be written as $\iota(A)=\mathrm{Ad}_U(A)=UAU^{-1}$ for some unitary $U\in \Ac$. 

Endomorphisms $\rho$ obtained from representations satisfying the DHR criterion for a given $\Ocal$ have the following properties: \cite{KlausSuperselection}
	\begin{enumerate}[label=\bf LE\arabic{enumi},leftmargin=*,widest=4] 
	\item \textbf{Localised} in $\Ocal$: $\rho(A)=A$, $A\in\Ac(\Ocal')$.\label{it:localised}
	\item \textbf{Transportable}: $\forall \Ocal_1,\Ocal_2$ with $\Ocal\cup\Ocal_1\subset \Ocal_2$, there is a unitary $U\in \Ac(\Ocal_2)$ with $\mathrm{Ad}_U\circ\rho(A)=A$, $A\in\Ac(\Ocal_1')$, i.e. for every region $\Ocal_1$ there exists an endomorphism equivalent to $\rho$  under \eqref{equiv} that is localized in $\Ocal_1$. \label{it:transp}
	\item $\rho(\Ac(\Ocal_1))\subset \Ac(\Ocal_1)$, $\forall \Ocal_1\supset\Ocal$.  In fact, this is a consequence of \ref{it:localised} and \ref{it:transp} by an argument similar to show $\tilde{\pi}(\Ac(\Ocal_1))\subset \pi_0(\Ac(\Ocal_1))$ above. \label{it:incl}
\end{enumerate}
Changing the perspective, we can use the properties above  
as defining properties for the following class of endomorphisms of $\Ac$:
\begin{df}
	Given a double-cone $\Ocal$, let $\Delta(\Ocal)$ be the set of all transportable endomorphisms of $\Ac$ localised in $\Ocal$, i.e.  those satisfying \ref{it:localised} and \ref{it:transp} (and hence \ref{it:incl}). We define
	 $\Delta\doteq \bigcup_\Ocal \Delta(\Ocal)$. 
\end{df}
It can be shown that the equivalence classes $\Delta/\sim$ (where $\sim$ is given by \eqref{equiv}) are in one to one correspondence with unitary equivalence classes of representations satisfying \eqref{DHRcrit} (i.e. with superselection sectors).  $\Delta/\sim$ is equipped with a natural product (given by the composition of endomorphisms of $\Ac$), which induces a product on the space of sectors, namely:
	\be\label{prod}
	[\pi_1\cdot \pi_2]\doteq [\pi_0\circ \rho_1\rho_2]\,,
	\ee
	where $\pi_i\doteq \pi_0\circ \rho_i$, $i=1,2$.
\begin{exercise}
	Show that the composition of representations introduced in \eqref{prod} is well defined and the resulting representation satisfies the DHR criterion.
\end{exercise}
It can also be shown (see e.g.~\cite{KlausSuperselection}) that the $\cdot$ product of two representations that are Poincar\'e covariant and satisfy the spectrum condition also has these two features. The space of sectors  equipped with the composition product is a semigroup with the vacuum sector as the identity. One can verify that simple sectors correspond to morphisms $\rho$ that are in fact automorphisms of $\Ac$, i.e. $\rho(\Ac)=\Ac$. Hence the space of simple sectors equipped with $\cdot$ is a group.

Transportability of endomorphisms \ref{it:transp} is crucial for the DHR analysis, since it allows us to ``move morphisms around''. First we establish the following: 
\begin{proposition}\label{prop:loccomm}
	Endomorphisms $\rho$ are locally commutative, i.e. for $\Ocal_1\subset\Ocal_2'$ and $\rho_i\in \Delta(\Ocal_i)$ we have $\rho_1\rho_2=\rho_2\rho_1$.
\end{proposition}
\begin{proof}
	Fix an arbitrary double cone $\Ocal$ and choose double cones $\hat{\Ocal}_i$, $\tilde{\Ocal}_i$, $i=1,2$ with the following properties:	 $\hat{\Ocal}_i\subset \Ocal'$, $\hat{\Ocal}_i\cup \Ocal_i\subset \tilde{\Ocal}_i$, $i=1,2$ and $\tilde{\Ocal}_1\subset \tilde{\Ocal}_2'$ (see Figure~\ref{fig:loccomm}). 
	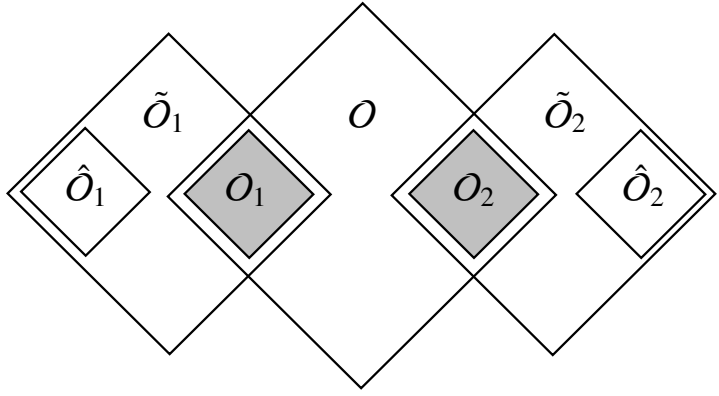
\begin{figure}
\centering
\begin{tikzpicture}[x=1.00mm, y=1.00mm, inner xsep=0pt, inner ysep=0pt, outer xsep=0pt, outer ysep=0pt,scale=0.6]
\path[line width=0mm] (20.34,15.35) rectangle +(178.75,89.08);
\definecolor{L}{rgb}{0,0,0}
\path[line width=0.30mm, draw=L] (100.16,102.44) [rotate around={225:(100.16,102.44)}] rectangle +(60.40,59.92);
\path[line width=0.30mm, draw=L] (142.16,95.61) [rotate around={225:(142.16,95.61)}] rectangle +(50.39,49.91);
\path[line width=0.30mm, draw=L] (57.62,95.70) [rotate around={225:(57.62,95.70)}] rectangle +(49.91,50.23);
\path[line width=0.30mm, draw=L] (39.33,74.86) [rotate around={225:(39.33,74.86)}] rectangle +(19.87,20.03);
\definecolor{F}{rgb}{0.753,0.753,0.753}
\path[line width=0.30mm, draw=L, fill=F] (75.18,74.37) [rotate around={225:(75.18,74.37)}] rectangle +(19.87,20.03);
\path[line width=0.30mm, draw=L, fill=F] (124.44,74.37) [rotate around={225:(124.44,74.37)}] rectangle +(19.87,20.03);
\path[line width=0.30mm, draw=L] (161.10,74.37) [rotate around={225:(161.10,74.37)}] rectangle +(19.87,20.03);
\draw(97,75.18) node[anchor=base west]{\fontsize{14.23}{17.07}\selectfont $\mathcal{O}$};
\draw(70,59) node[anchor=base west]{\fontsize{14.23}{17.07}\selectfont $\mathcal{O}_1$};
\draw(120,59) node[anchor=base west]{\fontsize{14.23}{17.07}\selectfont $\mathcal{O}_2$};
\draw(35,59) node[anchor=base west]{\fontsize{14.23}{17.07}\selectfont $\hat{\mathcal{O}}_1$};
\draw(157.15,59) node[anchor=base west]{\fontsize{14.23}{17.07}\selectfont $\hat{\mathcal{O}}_2$};
\draw(52,75.18) node[anchor=base west]{\fontsize{14.23}{17.07}\selectfont $\tilde{\mathcal{O}}_1$};
\draw(140,75.18) node[anchor=base west]{\fontsize{14.23}{17.07}\selectfont $\tilde{\mathcal{O}}_2$};
\end{tikzpicture}
\caption{Configuration of doublecones in the proof of Proposition~\ref{prop:loccomm}.\label{fig:loccomm}}
\end{figure}
	\ref{it:transp} implies that there exist unitaries $U_i\in\Ac(\tilde{\Ocal}_i)$ such that $\mathrm{Ad}_{U_i}\circ \rho_i=\hat{\rho}_i\in\Delta(\hat{\Ocal}_i)$, $i=1,2$. Hence for any $A\in\Ac(\Ocal)$ we have $\rho_i(A)= \mathrm{Ad}_{U_i^*}\circ\hat{\rho}_i(A)=\mathrm{Ad}_{U_i^*}(A)$, as $\Ocal$ is causally disjoint from both $\hat\Ocal_1$ and $\hat\Ocal_2$ and we have also used \ref{it:incl}. Hence
	\[
	\rho_1\rho_2(A)=\rho_1\circ \mathrm{Ad}_{U_2^*}(A)=\mathrm{Ad}_{\rho_1(U_2^*)}\circ \mathrm{Ad}_{U_1^*}(A)=\mathrm{Ad}_{U_2^*U_1^*}(A)=\mathrm{Ad}_{U_1^*U_2^*}(A)=\rho_2\rho_1(A)\,,
	\]
	where we used \ref{it:localised} to conclude that $\rho_1(U_2^*)=U_2^*$, while $U_1^*U_2^*=U_2^*U_1^*$ follows from \ref{it:Eins}. We can repeat the same reasoning for any double cone $\Ocal$, so $\rho_1\rho_2=\rho_2\rho_1$.
\end{proof}
We can now use this result to show that the product of representations is commutative, i.e. $[\pi_1\cdot \pi_2]=[\pi_2\cdot \pi_1]$. Let $\rho_1,\rho_2\in\Delta(\Ocal)$
and take two spacelike separated double cones $\Ocal_1$, $\Ocal_2$. By \ref{it:transp} there are morphisms $\tilde{\rho}_i$, $i=1,2$ localized in $\Ocal_i$ and unitaries $U_i$, $i=1,2$ such that $\tilde{\rho}_i=\mathrm{Ad}_{U_i}\circ\rho_i$. It follows by Prop.~\ref{prop:loccomm} that
\[
\rho_2\rho_1=\mathrm{Ad}_{\rho_2(U_1^*)U_2^*}\circ \tilde{\rho}_2\tilde{\rho}_1=\mathrm{Ad}_{\rho_2(U_1^*)U_2^*}\circ\tilde{\rho}_1\tilde{\rho}_2=
\mathrm{Ad}_{\varepsilon(\rho_1,\rho_2)}\circ \rho_1\rho_2\,,
\]
with the unitary 
\be\label{eq:statop}
\varepsilon(\rho_1,\rho_2)=\rho_2(U_1^*)U_2^*U_1\rho_1(U_2)\in\Ac\,.
\ee
This proves the equivalence of $\rho_2\rho_1$ and $\rho_1\rho_2$, so the commutativity of the product of representations follows. Hence the space of all sectors, equipped with $\cdot$ is an abelian semigroup and the space of all simple sectors equipped with $\cdot$ is an abelian group.

The unitary operator $\varepsilon(\rho_1,\rho_2)$ that we have discovered here is called the \emph{statistics operator}. It depends only on $\rho_1$, $\rho_2$ and not on $\tilde{\rho}_i$ nor the choice of $U_i$, $i=1,2$. 

Let $\varepsilon_\rho\equiv\varepsilon(\rho,\rho)$. In spacetime dimension $d>2$ we have $\varepsilon_\rho^2=\id$. For simple sectors $\pi_0\circ \rho^2$ is irreducible, so $\varepsilon_\rho$ is a multiple of $\1$ and in $d>2$ this implies that $\varepsilon_\rho=\pm \1$. Physically this corresponds to the alternative between Fermi and Bose statistics. For general sectors in $d>2$ one can also have para-statistics or even infinite statistics. In lower dimensions ($d\leq 2$), $\varepsilon_\rho$ is instead related to representations of the \emph{braid group}. This kind of behaviour appears also for cone-localized charges (BF analysis), but in this case the braided statistics appears already in $d\leq 3$. More details relating $\varepsilon_\rho$ to statistics will be given in the next section. 

The physical interpretation of the product of representations can be understood as follows \cite{DHR3}. Let $\omega_0$ be the vacuum state and $\rho_i\in\Delta(\Ocal_i)$, $i=1,2$. We have the two states $\omega_i=\omega_0\circ \rho_i$, which are vector states in their respective representations $\pi=\pi_0\circ \rho_i$; indeed, one suitable vector is the vacuum vector $\Omega$ that is the GNS vector induced by $\omega_0$. Clearly, $\omega=\omega_0\circ \rho_1\rho_2$ is a vector state in representation $[\pi_1\cdot \pi_2]$, with the same representing vector. Let
$\Ocal_1$ be spacelike to $\Ocal_2$. Then $\omega$
 looks like the vector state $\omega_0\circ\rho_1$ with respect to observations in $\Ocal_2'$ and like $\omega_0\circ\rho_2$ for observations in $\Ocal_1'$. Hence, taking the product $[\pi_1\cdot \pi_2]$ has a physical interpretation of composing two vector states that are localized ``far apart''.
\subsection{Intertwiners and permutation symmetry}\label{sec:intertwiners}
As stated in the previous section, states $\omega_0\circ \rho_k$, $k=1,\dots,n$ obtained from transportable localized morphisms $\rho_k\in \Delta$ are interpreted as vector states corresponding to localized charged particles and composition of morphisms describes creating several such charges in spacetime. Now we want to understand how one may create these charges in a given order. As the state $\omega_0\circ\rho_1\dots \rho_n$ is independent of the ordering of the $\rho_k$'s, this must be done by permuting a family of vectors all of which representing this state. In order to do this, we need some notation.

\begin{df}
	Let $\rho,\sigma\in\Delta$. Define the space of \emph{intertwiners} between $\rho$ and $\sigma$ by\footnote{In \cite{DHR3,DHR4} the intertwiners are denoted by $\mathbf{T}\equiv (\sigma|T|\rho)$, to emphasize to which space they belong.}
	\[
	(\sigma,\rho)\doteq \{T\in\Ac|\ \sigma(A) T=T\rho(A),~\forall~A\in\Ac\}\,.
	\]
\end{df}
We can now define some algebraic operations on intertwiners. Let $S\in (\tau,\sigma)$ and $T\in (\sigma,\rho)$ be intertwiners. We can compose them to obtain $ST\in(\tau,\rho)$ and can also define the adjoint $T^*\in(\rho,\sigma)$. Moreover, there is a natural product between intertwiners in different spaces. Let $T_1\in (\sigma_1,\rho_1,)$ and $T_2\in (\sigma_2,\rho_2)$. We define $ T_1\times T_2\in (\sigma_1\sigma_2,\rho_1\rho_2)$ as
\[
T_1\times T_2\doteq T_1\rho_1(T_2)=\sigma_1(T_2)T_1
\]
(using the fact that $T_1\in (\sigma_1,\rho_1)$ for the last equality).
The $\times$-product is associative, and distributive with respect to the composition of intertwiners (i.e., their product in $\Ac$).

Next we discuss localization properties of the intertwiners.
Consider $T\in(\sigma,\rho)$, where $\rho$ is supported in $\Ocal_1$ and $\sigma$ in $\Ocal_2$. We call $\Ocal_2$ the left support and $\Ocal_1$ the right support of the intertwiner $T$. 
If $A\in \Ac(\Ocal_1')\cap\Ac(\Ocal_2')$, then the support and intertwining properties imply
\[
TA = T\rho(A)= \sigma(A)T = AT
\]
so $T$ is \emph{bilocal}, in the sense that $T\in(\Ac(\Ocal_1')\cap\Ac(\Ocal_2'))'$ (here the commutant is taken in $\Ac$). In the case $\Ocal_1=\Ocal_2=\Ocal$, we have $T\in\Ac(\Ocal)$ by Haag duality~\eqref{Haagdual3} in the vacuum representation and the assumption that $\pi_0$ is faithful.
 
 We call two intertwiners \textit{causally disjoint} if their right supports lie space-like to each other and the same holds for their left supports. 
\begin{proposition}\label{prop:intcomm}
Let	 $T_i\in(\sigma_i,\rho_i)$, where $\rho_i,\sigma_i\in\Delta(\Ocal_i)$, $i=1,2$, with $\Ocal_1$ and $\Ocal_2$ spacelike separated (i.e. $T_1$ and $T_2$ are causally disjoint), then
\[
T_1\times T_2=T_2\times T_1\,.
\]
\end{proposition}
 \begin{proof}
 	To see this, note that $T_i\in\Ac(\Ocal_i)$, so $T_1\times T_2=T_1\rho_1(T_2)=T_1T_2=\sigma_2(T_1)T_2=T_2\times T_1$.
 \end{proof}
The statistics operator $\varepsilon_\rho$ is also an intertwiner and it has a nice expression in terms of products of other intertwiners.
\begin{exercise} Let $\rho_1,\rho_2\in\Delta(\Ocal)$ and $\tilde{\rho}_i=U_i\rho_i U_i^{-1}$, $i=1,2$, as in the construction leading to \eqref{eq:statop}. Show that $\varepsilon(\rho_1,\rho_2)\in(\rho_2\rho_1,\rho_1\rho_2)$ and 
	\[
	\varepsilon(\rho_1,\rho_2)=(U_2^*\times U_1^*)  (U_1\times U_2)\,.
	\]
\end{exercise}
Clearly, 
$\varepsilon_\rho\in(\rho^2,\rho^2)$, which means that $\varepsilon_\rho$ commutes in $\Ac$ with every element of $\rho^2\Ac$, $\varepsilon_\rho\in (\rho^2\Ac)'$, and therefore $\pi_0(\varepsilon_\rho)$ commutes with all observables in the representation $\pi_0\circ\rho^2$.

Now take $T_i\in (\sigma_i,\rho_i)$, where $\rho_i,\sigma_i\in\Delta$, $i=1,2$. One can also easily check (\textit{Exercise}) that
\be\label{eq:changeorder}
\varepsilon(\sigma_1,\sigma_2)  (T_1\times T_2)=(T_2\times T_1)  \varepsilon(\rho_1,\rho_2)\,,
\ee
so in particular, for $\sigma_1=\sigma_2$ and $\rho_1=\rho_2$, we have
\[
\varepsilon_\sigma \, (T_1\times T_2)=(T_2\times T_1)  \varepsilon_\rho\,.
\]
Hence $\varepsilon$ changes the order of factors in $\times$-product.
\begin{exercise}
	Prove \eqref{eq:changeorder}. You will need to use Proposition~\ref{prop:intcomm}.
\end{exercise}

Using $\varepsilon_\rho$, we can construct a representation $\varepsilon^{(n)}_\rho$ of the \emph{braid group} for each $\rho$. The braid group $B_n$ with $n$ strands is the group generated by $\varsigma_1,\dots \varsigma_{n-1}$ with relations:
\begin{align*}
\varsigma_i\varsigma_j&=\varsigma_j\varsigma_i\,,\quad \textrm{if}\ |i-j|>2\,,\\
\varsigma_i \varsigma_{i+1} \varsigma_i&=\varsigma_{i+1}\varsigma_i \varsigma_{i+1}\,.
 \end{align*}
Let $\varsigma_i$, $i\in\{1,\dots,n-1\}$ be a generator of  $B_n$, then we set $\varepsilon^{(n)}_\rho(\varsigma_i)\doteq \rho^{i-1}\varepsilon_\rho$. 
The following exercise shows that this is indeed a representation of the braid group.
\begin{exercise}
   Show that $\varepsilon_\rho\rho(\varepsilon_\rho)=\varepsilon(\rho^2,\rho)$ and use this together with \eqref{eq:changeorder} to prove $\varepsilon_\rho\rho(\varepsilon_\rho)\varepsilon_\rho=\rho(\varepsilon_\rho)\varepsilon_\rho\rho(\varepsilon_\rho)$.
\end{exercise}

In $d>2$, this also gives us a representation of the permutation group $S_n$ (as $\varepsilon_\rho^2=\1$). In the more detailed analysis that follows we will focus on the $d>2$ case and we refer the reader to \cite{FRS89} for the general case. Let  $\varepsilon^{(n)}_\rho(P)$ denote the representative of $P\in S_n$. More generally, we may also define $\varepsilon(\rho_1,\dots,\rho_n; P)\in (\rho_{P^{-1}(1)}\dots \rho_{P^{-1}(n)},\rho_1\dots\rho_n)$ by
\[
\varepsilon(\rho_1,\dots,\rho_n;P):=\boldsymbol{U}^*(P)  \boldsymbol{U}(e)\,,
\]
where (in analogy with the $n=2$ case) $\tilde{\rho}_i=U_i\rho_i U_i^{-1}$ are auxiliary morphisms localized in spacelike separated regions $\Ocal_i$, $i=1\dots,n$ and we use the notation $\boldsymbol{U}(P):=U_{P^{-1}(1)}\times\dots \times U_{P^{-1}(n)}$. One can easily check (\textit{Exercise}) that $\varepsilon(\rho,\dots,\rho;P)=\varepsilon^{(n)}_\rho(P)$. Note that in $n=2$ case we implicitly have $\varepsilon(\rho_1,\rho_2;\tau)\equiv \varepsilon(\rho_1,\rho_2)$, where $\tau$ is the transposition of $1$ and $2$.
Property \eqref{eq:changeorder} generalizes to:
\[
\varepsilon(\sigma_1\dots,\sigma_n;P)   \boldsymbol{T}(e)=\boldsymbol{T}(P)  \varepsilon(\rho_1,\dots,\rho_n;P)\,,
\]
where $T_i\in(\sigma_i,\rho_i)$  $i=1,\dots,n$. If $\sigma_i=\sigma$, $\rho_i=\rho$ for all $i=1,\dots,n$, we have
\[
\varepsilon^{(n)}_\sigma(P)   \boldsymbol{T}(e)=\boldsymbol{T}(P)  \varepsilon^{(n)}_\rho(P)\,.
\]

To understand better the physical interpretation of $\varepsilon^{(n)}_\rho(P)$, fix a morphism $\rho$ and consider a family of intertwiners $U_k\in(\sigma_k,\rho)$, where $\sigma_k$ are morphisms localized in spacelike separated regions $\Ocal_k$, $k=1,\dots,n$. Let $\Omega\in\Hcal_0$ be 
the vacuum vector. Clearly,
\[
\left<\pi_0(U_k^*)\Omega,\pi_0(\rho(A)) \pi_0(U_k^*)\Omega\right>=\left<\Omega,\pi_0(U_k\rho(A)U_k^*)\Omega\right>=\left<\Omega,\pi_0(\sigma_k(A)) \Omega\right>=\omega_0\circ \sigma_k(A)\,,
\]
so $\omega_0\circ\sigma_k$ is a vector state on $\rho\Ac$ with the distinguished vector $\pi_0(U_k^*)\Omega$.

Now consider vectors of the form
\[
\Psi_P=\pi_0(\boldsymbol{U}^*(P))\Omega\,, \qquad P\in S_n,
\]
which can be interpreted as a product of $n$ state vectors with
identical charge quantum numbers but with an ordering determined by $P$. The operator  $\varepsilon^{(n)}_\rho(Q)$ changes the order of factors in this vector \cite{DHR3}. To see this, recall that
$Q\mapsto \varepsilon^{(n)}_\rho(Q)\in\Ac$ is a unitary representation of $S_n$, whereupon $\Psi_P=\pi_0( \varepsilon^{(n)}_\rho(P)  \boldsymbol{U}^*(e))\Omega$ and
\[\pi_0( \varepsilon^{(n)}_\rho(Q))\Psi_P=\pi_0(\varepsilon^{(n)}_\rho(Q) \varepsilon^{(n)}_\rho(P)\boldsymbol{U}^*(e))\Omega = \pi_0(\varepsilon^{(n)}_\rho(QP)\boldsymbol{U}^*(e))\Omega=\Psi_{QP}
\]
Maybe: Each vector $\Psi_P$ induces the same state on $\rho^n\Ac$, namely $\omega_0\circ \sigma_1\dots\sigma_n$. Therefore, the action of $\varepsilon^{(n)}_Q$ is analogous to permutations of the wave functions of $n$ identical particles in quantum mechanics, which also leaves expectation values of observable quantities unchanged.

In the next step one introduces the notion of a \emph{conjugate sector}. Physically, the relation between a sector and its conjugate is that of having a charged particle localized in some compact region versus having the corresponding antiparticle localized in that region. We have already mentioned that sectors can be equipped with the structure of a semigroup and simple sectors form a group. In the latter case, the conjugate sector is just given by the group inverse $\varepsilon_\rho=\pm \1$, so we have the simple fermion/boson alternative.

More generally, to  obtain a left inverse to a given $\rho$, we want to find a map $\phi:\Ac\rightarrow \Ac$ with
\[
\phi(\rho(A)B\rho(C))=A\phi(B)C\,,\quad \phi(A^*A)\geq0\,,\quad \phi(\1)=\1\,.
\]
Note that $\phi$ on $\rho(\Ac)$ can be set as $\rho^{-1}$. We then use the Hahn-Banach theorem to extend the state $\omega_0\circ \rho^{-1}$ on $\rho(\Ac)$ to $\Ac$. Let $(\overline{\pi},\overline{\Hcal}, \overline{\Omega})$ be the corresponding GNS triple and define an isometry $V:\Hcal_{0}\rightarrow \overline{\Hcal}$, by means of
$VA\Omega=\overline{\pi}\circ \rho(A)\tilde{\Omega}$. We then define the left inverse of the given sector $\rho$ as $\phi(A)=V^* \overline{\pi}(A) V$. 

The left inverse is used to study the representations of the permutation
group in $d>2$ (or the braid group in low dimensions). We note that
\[
\phi(\varepsilon_\rho)\rho(A)=\phi(\varepsilon_\rho\rho^2(A))=\phi(\rho^2(A)\varepsilon_\rho)=\rho(A)\phi(\varepsilon_\rho)\,,
\]
which for irreducible $\rho$ implies
\[
\phi(\varepsilon_\rho)=\lambda_\rho \1\,,
\]
where $\lambda_\rho$ is the statistics parameter and it characterizes the statistics of the sector $\rho$.  One finds that the allowable values for this parameter are:  
\begin{itemize}
\item $\lambda=\frac{1}{d}$, $d\in\NN$ giving para Bose statistics of order $d$;
\item $\lambda=-\frac{1}{d}$, $d\in\NN$ is giving para Fermion statistics of order $d$; 
\item $\lambda=0$, giving infinite statistics. 
\end{itemize}
In DHR theory, $d$ is called \emph{the statistical dimension}.\footnote{Mathematically, as discovered by Longo in \cite{Lon89}, $d$ is in fact the square root  of the Jones index of the inclusion $\rho(\Ac(\Ocal))\subset\rho(\Ac(\Ocal'))'$, i.e. it quantifies how badly is the Haag duality broken in the given sector.}

To obtain the sector $\overline{\rho}$ conjugate to a given sector $[\rho]$, one shows (under appropriate, physically motivated assumptions, including the finite statistics $\lambda\neq 0$ \cite{DHR3}) that $\overline{\pi}$ satisfies the DHR criterion, so there exists a  morphism $\overline{\rho}\in \Delta(\Ocal)$ such that $\overline{\pi}=\pi_0\circ \overline{\rho}$ and an isometry $R$ such that $\overline{\rho}\rho(A) R= RA$.

In fact, in the language of category theory, localized morphisms and intertwiners form a symmetric (or braided in $d<2$) monoidal category, where the monoidal structure is given by $\times$. The existence of conjugate sectors (unique up to equivalence) means that the category is rigid.

 The reconstruction theorem of Doplicher and Roberts \cite{DR89} allows one to reconstruct the field net $\F(\Ocal)$ and the gauge group $G$ from the above data completing the programme set out in Sec.~\ref{interest}. Abstractly, firstly they show the equivalence of the DHR category to a category of representations of some compact group and then reconstruct the group from that category (this step is a version of the Tannaka--Krein duality), together with the algebra of fields on which this group acts.

\section{Conclusions}
In these notes, we have summarized some important aspects of quantum field theory in the algebraic formulation of Haag and Kastler, focusing on the features that make it very different from quantum mechanics. One such feature is the existence of inequivalent representations. This was illustrated by the case study of the van Hove model in Section~3, and further emphasized in Section~8, where inequivalent representations corresponding to different charges were discussed. In both situations, one can see that all the physical information can be recovered from the abstract net of algebras, so it is more advantageous to think of the net rather than the collection of Hilbert space representations as the fundamental object. The axioms for the net have been formulated in Section~4.1, followed by some simple examples in Section~4.2. The connection between the algebraic and the Hilbert space centred approaches can be made through the choice of an algebraic state. In particular, for QFT on Minkowski spacetime\footnote{Here we focused only on QFT on Minkowski spacetime, but the algebraic approach also easily generalizes to curved spacetimes \cite{BFV,FewVerchReview,JMPReview}.}, one can consider the distinguished Poincar{\'e} invariant state, \textit{the vacuum}; we have also described the more general class of quasi-free states that can be specified by their two-point functions and have Fock space representations. In Section~5, we showed that QFT in the vacuum representation has some peculiar features that make it very different from quantum mechanics, the most dramatic being the Reeh-Schlieder theorem. 

A net of $C^*$-algebras together with a state induces a net of von Neumann algebras via weak completion. In Section~6 we pointed out that the type of von Neumann algebras that arise in QFT (type III) is very different from the type characteristic for quantum mechanics (type I). We discussed the main consequences of this fact in Section~7, in the context of independence of measurements by spacelike observers, in the guise of the split property.

To close, let us emphasize that AQFT, although very different from quantum mechanics, is not
to be regarded as disjoint from ``traditional'' QFT as presented in standard textbooks. To the contrary, it is a framework that allows one to derive and study common structural and conceptual features of QFT, which then become realized in physical, experimentally testable models.



\appendix
\section{Some basic functional analysis}\label{appx:basic_fa}

A general reference for this brief summary is~\cite{ReedSimon:vol1}.
Recall that a Hilbert space $\HH$ is a complex inner product space, with an inner product $\ip{\cdot}{\cdot}$ that is linear in the second slot and conjugate-linear in the first, and for which the associated norm $\|\psi\|:=\sqrt{\ip{\psi}{\psi}}$ is complete, i.e., all Cauchy sequences converge. The Hilbert space is \emph{separable} if it has a finite or countably infinite orthonormal basis, and \emph{inseparable} otherwise.  

A linear operator $A:\HH_1\to\HH_2$ between Hilbert spaces $\HH_1$ and $\HH_2$ is \emph{bounded} iff $\|A\|\doteq \sup_{\|x\|_1=1} \|Ax\|_2$ is finite, where we use the subscript to denote the Hilbert space norm concerned. For maps between Hilbert spaces, boundedness and continuity are equivalent properties. In this case $A$ has an \emph{adjoint} $A^*:\HH_2\to\HH_1$, with the defining property 
\be\label{adjoint}
\ip{\ph}{A\psi}_2=\ip{A^*\ph}{\psi}_1\,,
\ee
for all $\psi\in\HH_1$, $\ph\in\HH_2$. Several interesting classes of bounded operator may be defined: a bounded operator $A:\HH\to\HH$ is \emph{self-adjoint} if $A=A^*$; while $A$ is a projection if $A=A^*=A^2$ (more strictly, this defines an `orthogonal projection' but we follow common usage in simply saying `projection'). A bounded operator $U:\HH_1\to\HH_2$ is \emph{unitary} if $U^*U = UU^* = \II$, and a \emph{partial isometry} if $U^*U$ and $UU^*$ are projections.  Every bounded operator $A:\HH\to\HH$ has a unique \emph{polar decomposition} $A = U|A|$ such that $U$ is a partial isometry with $\ker U=\ker A$, and $|A|$ is a positive operator such that $|A|^2 = A^*A$. Here, a self-adjoint operator $A:\HH\to\HH$ is said to be \emph{positive} if $\ip{\psi}{A\psi}\ge 0$ for all $\psi\in\HH$. For obvious reasons $|A|$ is called the positive square root of $A^*A$. An operator $A$ on $\HH$ is said to be of \emph{trace class} if
$\sum_{\alpha} \ip{e_\alpha}{\,|A|e_\alpha}$ is finite, where $e_\alpha$ is some
orthonormal basis of $\HH$; in this case, the trace $\tr A := \sum_{\alpha} \ip{e_\alpha}{A e_\alpha}$ is finite and independent of the basis used to compute it. 

If $A$ is a partially defined linear map between Hilbert spaces $\HH_1$ and $\HH_2$, 
we denote its domain of definition within $\HH_1$ by $D(A)$. We typically only consider the situation where $D(A)$ is dense. If $\sup\{\|Ax\|_2:x\in D(A),~\|x\|_1=1\}$ is finite, then $A$ can be extended by continuity to a unique bounded operator from $\HH_1$ to $\HH_2$; otherwise, $A$ is described as an \emph{unbounded operator}. The adjoint $A^*$ of a densely defined unbounded operator $A$ is again defined through \eqref{adjoint} and $D(A^*)$ is the set of all $\ph\in\Hcal$ for which this definition makes sense: $\varphi\in D(A^*)$ if and only if there exists $\eta\in\Hcal$ such that $\ip{\ph}{A\psi}=\ip{\eta}{\psi}$ holds for all $\psi\in D(A)$, whereupon we write $A^*\varphi=\eta$. An unbounded densely defined operator $A$ is called self-adjoint if $\ip{\ph}{A\psi}=\ip{A\ph}{\psi}$ for all $\ph,\psi\in D(A)$ and in addition $D(A^*)= D(A)$.
Any operator is completely described by its graph 
\[
\Gamma(A):=\{(x,Ax)\in \HH\times\HH: x\in D(A)\}.
\]
One says that $A$ is a \emph{closed operator} if $\Gamma(A)$ is a closed subset of $\HH\times\HH$ with respect to the norm of $\HH\oplus\HH$; it is \emph{closable} if $\Gamma(A)$ has a closure that is the graph of some (closed) operator $\overline{A}$, which is naturally called the \emph{closure} of $A$. All densely defined operators with densely defined adjoints are closable and their adjoints are closed, so in particular self-adjoint operators are closed. The polar decomposition extends to closed operators. 
 
The \emph{spectrum} $\sigma(A)$ of a (bounded or unbounded) operator $A$ on a Hilbert space $\HH$ is the set of $z\in\CC$ for which $A-z\II$ fails to have a bounded two-sided inverse. 
In particular, eigenvalues lie in the spectrum but not every spectral point is an eigenvalue.
The spectrum of a self-adjoint operator is real, $\sigma(A)\subset\RR$, while the spectrum of a unitary operator lies on the unit circle in $\CC$.  

Finally, suppose that $\HH$ is a \emph{real} Hilbert space (i.e., a real inner product space with a complete induced norm). To distinguish real spaces from complex ones in this appendix, we write the real inner product with round brackets and the adjoint with a dagger. 
A \emph{complex structure} on $\HH$ is a linear map $J:\HH\to \HH$ obeying $J^2 = -\II$, $J^\dagger = -J$. Then we may convert $\HH$ into a complex Hilbert space by adding two structures: first, the operation of
multiplication by a complex scalar,
\[
\CC\times\HH\owns (z,\psi)\mapsto  (\Re z)\psi - (\Im z)J\psi \in\HH ,
\]
in which sense multiplication by $i$ is implemented by $-J$ (this convention is annoying but avoids a proliferation of minus signs elsewhere in the main body of the text), and second, 
a sesquilinear inner product
\[
\ip{\psi}{\varphi} = (\psi,\varphi) + i(\psi,J\varphi).
\]
\begin{exercise}
	Check that $\HH$, with these additional structures, is indeed a complex Hilbert space.
\end{exercise} 

\section{Construction of an algebra from generators and relations} \label{appx:presentation}

Several algebras encountered in Sec.~\ref{sec:examples} were presented in terms of generators and relations. Here, we give more details on how an algebra may be constructed in this way, taking the real scalar field as our example. 
		\begin{itemize}
			\item First consider the free unital $*$-algebra $\Uc$ containing arbitrary finite linear combinations of finite products of the $\Phi(f)$'s and $\Phi(f)^*$'s and unit $\II$.
			\item Construct a two-sided $*$-ideal $\Ic$ in $\Uc$ generated by the relations. 
			Thus $\mathcal{I}$ contains all finite linear combinations of terms of the form
			\[
			A(\Phi(f)^*-\Phi(\overline{f}))B 
			\] 
			as $A$ and $B$ range over $\Uc$ and $f$ ranges over $\CoinX{\MM}$, and similar terms obtained from the other relations, and all terms obtained from these by applying $*$. Recall that a two-sided $*$-ideal is a subspace of the algebra that is stable under multiplication by algebra elements on either side, and under the $*$-operation.
			\item The algebra $\Ac(\MM)$ is defined as the quotient $\Uc/\Ic$, namely, the vector space quotient, equipped with product and $*$-operations so that
			\[ 
			[A] [B]= [AB],\qquad [A]^*=[A^*], \qquad \II_{\Ac(\MM)} = [\II_\Uc].
			\]
			One may check that the fact that $\Ic$ is a two sided $*$-ideal guarantees that these operations are well-defined (independent of the choice of representatives).
			
			\item For future reference: let $\Bc$ be another algebra obtained as a quotient $\Bc=\Vc/\Jc$, where $\Jc$ is a two-sided $*$-ideal in a unital $*$-algebra $\Vc$. Then
			any function mapping the $\Phi(f)$ into $\Vc$ extends uniquely to a unit-preserving $*$-homomorphism from $\Uc$ to $\Vc$, and induces a unit-preserving $*$-homomorphism from $\Ac(\MM)$ to $\Bc$, provided every element of $\Ic$ is mapped into $\Jc$, i.e., the map on generators is compatible with the relevant relations.
		\end{itemize}

\section{Fock space}\label{sec:Fock}

Let $\HH$ be a complex Hilbert space. As usual, $\HH^{\odot n}$ denotes the $n$'th symmetric tensor power of $\HH$, with $\HH^{\odot 0}=\CC$, whereupon the \emph{bosonic Fock space over $\HH$} is
\[
\FF(\HH) = \bigoplus_{n=0}^\infty \HH^{\odot n}.
\]
Thus, a typical Fock space vector is a sequence $\Psi = (\Psi_n)_{n\in\NN_0}$, where
$\Psi_n\in\HH^{\odot n}$ is called the $n$-particle component of $\Psi$. 
In particular the Fock vacuum vector is $\Omega=(1,0,\ldots)$, and the number operator
$N$ is defined by 
\[
(N\Psi)_n = n\Psi_n \qquad n\in\NN_0
\]
on the domain of all $\Psi\in\FF(\HH)$ for which $\sum_{n=0}^\infty n^2\|\Psi_n\|^2<\infty$.
We will describe the annihilation and creation operators and the number operator on $\FF(\HH)$ in the basis-free notation used in Sec.~\ref{sec:quasifree}. 
See references~\cite[\S 5.2.1]{BratRob:vol2} and~\cite[\S X.7]{ReedSimon:vol2} for more details.

In this framework, each $\psi\in\HH$ labels annihilation and creation operators $a(\psi)$ and $a^*(\psi)$ on $\FF(\HH)$, with $\psi\mapsto a(\psi)$ being \emph{antilinear}, and $\psi\mapsto a^*(\psi):=a(\psi)^*$ being linear in $\psi$. These operators are unbounded and have to be defined on suitable dense domains within $\FF(\HH)$, which may be taken as the domain of $N^{1/2}$, i.e., those $\Psi\in\FF(\HH)$ for which $\sum_{n=0}^\infty n\|\Psi_n\|^2<\infty$. The annihilation operator acts (on vectors in the domain) by
\[
(a(\varphi)\Psi)_n = \sqrt{n+1} \ell_{n+1}(\varphi)(\Psi_{n+1})
\]
where $\ell_{n+1}(\varphi):\HH^{\otimes (n+1)}\to\HH^{\otimes n}$ is defined by 
\[
\ell_{n+1}(\varphi)(\psi_1\otimes\cdots\otimes\psi_{n+1}) = \ip{\varphi}{\psi_1}\psi_2\otimes\cdots\otimes \psi_{n+1}
\]
and restricts to a map $\HH^{\odot (n+1)}\to\HH^{\odot n}$. It follows in particular that $a(\varphi)\Omega=0$ for all $\varphi\in\HH$. One may check that the adjoint operators obey
\[
(a^*(\varphi)\Psi)_0=0,\qquad
(a^*(\varphi)\Psi)_{n+1} = \sqrt{n+1} S_{n+1} (\varphi\otimes\Psi_{n}),\qquad n\in\NN_0
\]
where $S_{n+1}$ is the orthogonal projection onto $\HH^{\odot (n+1)}$ in $\HH^{\otimes (n+1)}$. 
Acting on vectors $\Psi$ in the domain of $N$, the canonical commutation relations hold in the form
\begin{equation}\label{eq:CCRsappx}
[a(\psi),a^*(\varphi)]\Psi = \ip{\psi}{\varphi}_\HH \Psi,
\end{equation}
and vectors obtained by acting with sums of products of $a^*(\varphi)$ operators on $\Omega$ are dense in $\FF(\HH)$.

This presentation of the annihilation and creation operators may seem unfamiliar to those who prefer their annihilation and creation operators to look more like $a_i$ and $a_j^*$. But let $(e_i)$ be any orthonormal basis for $\HH$ and define $a_i = a(e_i)$, $a_i^*=a(e_i)^*=a^*(e_i)$. Then the CCRs become
\[
[a_i,a_j^*] = \ip{e_i}{e_j}_\HH \II = \delta_{ij}\II
\]
(understood as acting on a suitable domain) and of course $a_i\Omega=a(e_i)\Omega=0$,
which provides a set of annihilation and creation operators labelled by a discrete index. At least formally (because infinite sums of unbounded operators should be handled with care)
\[
a(\psi) = \sum_i \ip{\psi}{e_i}a_i,\qquad 
a^*(\varphi) = \sum_i \ip{e_i}{\varphi}a_i^*.
\]
The advantage of the basis-independent approach is that it does not give any basis a privileged status, and avoids the need for infinite series of the type just given if changing basis, for example. The number operator can also be related to the annihilation and operators in the basis-independent form -- see~\cite[\S 5.2.3]{BratRob:vol2}.

It is also common in QFT to use annihilation and creation operators indexed by a continuous momentum variables, in situations where $\HH$ is a space of square-integrable functions of momentum. This is essentially a matter of using a continuuum-normalised `improper basis' for $\HH$, but one should be aware that, while
$a(\kb)$ does define an (unbounded) operator, it is sufficiently poorly-behaved that it does not have a densely defined operator adjoint. Nonetheless, any normal-ordered string 
$a^*(\kb_1)\cdots a^*(\kb_m)a(\kb'_{1})\cdots a(\kb'_n)$ can be given
meaning as a quadratic form, that is, defining its matrix elements as
\[
\ip{\Psi}{a^*(\kb_1)\cdots a^*(\kb_m)a(\kb'_{1})\cdots a(\kb'_n)\Psi'}:=
\ip{a(\kb_m)\cdots a(\kb_1)\Psi}{a(\kb'_{1})\cdots a(\kb'_n)\Psi'}
\]  
on suitable vectors $\Psi,\Psi'\in\FF(\HH)$. For more on this viewpoint, see~\cite[\S X.7]{ReedSimon:vol2}.

\providecommand{\bysame}{\leavevmode\hbox to3em{\hrulefill}\thinspace}
\providecommand{\MR}{\relax\ifhmode\unskip\space\fi MR }
\providecommand{\MRhref}[2]{%
	\href{http://www.ams.org/mathscinet-getitem?mr=#1}{#2}
}
\providecommand{\href}[2]{#2}
\small{}

\begin{thebibliography}{BDFY15}
		
		\bibitem[Ara99]{Araki}
		Huzihiro Araki, \emph{Mathematical theory of quantum fields}, International
		Series of Monographs on Physics, vol. 101, Oxford University Press, New York,
		1999, Translated from the 1993 Japanese original by Ursula Carow-Watamura.
		\MR{1799198}
		
		\bibitem[BDF87]{BucDAnFre:1987}
		D.~Buchholz, C.~D'Antoni, and K.~Fredenhagen, \emph{The universal structure of
			local algebras}, Comm. Math. Phys. \textbf{111} (1987), 123--135. \MR{896763
			(88j:46055)}
		
		\bibitem[BDFY15]{AdvAQFT}
		Romeo Brunetti, Claudio Dappiaggi, Klaus Fredenhagen, and Jakob Yngvason
		(eds.), \emph{Advances in algebraic quantum field theory}, Mathematical
		Physics Studies, Springer International Publishing, 2015.
		
		\bibitem[BF82]{BF82}
		D.~Buchholz and K.~Fredenhagen, \emph{Locality and the structure of particle
			states}, Commun. Math. Phys. \textbf{84} (1982), no.~1, 1--54.
		
		\bibitem[BFV03]{BFV}
		R.~Brunetti, K.~Fredenhagen, and R.~Verch, \emph{The generally covariant
			locality principle---{A} new paradigm for local quantum field theory},
		Commun. Math. Phys. \textbf{237} (2003), 31--68.
		
		\bibitem[BGP07]{BarGinouxPfaffle}
		Christian B{\"{a}}r, Nicolas Ginoux, and Frank Pf{\"{a}}ffle, \emph{Wave
			equations on {L}orentzian manifolds and quantization}, European Mathematical
		Society (EMS), Z{\"{u}}rich, 2007.
		
		\bibitem[BR87]{BratRob:vol1}
		Ola Bratteli and Derek~W. Robinson, \emph{Operator algebras and quantum
			statistical mechanics. 1}, second ed., Texts and Monographs in Physics,
		Springer-Verlag, New York, 1987, $C^\ast$- and $W^\ast$-algebras, symmetry
		groups, decomposition of states. \MR{887100}
		
		\bibitem[BR97]{BratRob:vol2}
		\bysame, \emph{Operator algebras and quantum statistical mechanics. 2}, second
		ed., Texts and Monographs in Physics, Springer-Verlag, Berlin, 1997,
		Equilibrium states. Models in quantum statistical mechanics. \MR{1441540}
		
		\bibitem[BR14]{BR14}
		D.~Buchholz and J.~E. Roberts, \emph{New light on infrared problems: sectors,
			statistics, symmetries and spectrum}, Communications in Mathematical Physics
		\textbf{330} (2014), no.~3, 935--972.
		
		\bibitem[DDFL87]{DAnDopFreLon:1987}
		C.~D'Antoni, S.~Doplicher, K.~Fredenhagen, and R.~Longo, \emph{Convergence of
			local charges and continuity properties of {$W^*$}-inclusions}, Comm. Math.
		Phys. \textbf{110} (1987), 325--348. \MR{888004 (88h:81078)}
		
		\bibitem[DHR69a]{DHR1}
		S.~Doplicher, R.~Haag, and J.~E. Roberts, \emph{Fields, observables and gauge
			transformations. {I}}, Commun. Math. Phys. \textbf{13} (1969), no.~1, 1--23.
		
		\bibitem[DHR69b]{DHR2}
		\bysame, \emph{Fields, observables and gauge transformations. {II}}, Commun.
		Math. Phys. \textbf{15} (1969), no.~3, 173--200.
		
		\bibitem[DHR71]{DHR3}
		\bysame, \emph{Local observables and particle statistics {I}}, Commun. Math.
		Phys. \textbf{23} (1971), no.~3, 199--230.
		
		\bibitem[DHR74]{DHR4}
		\bysame, \emph{Local observables and particle statistics {II}}, Commun. Math.
		Phys. \textbf{35} (1974), no.~1, 49--85.
		
		\bibitem[DL84]{DopLon:1984}
		S.~Doplicher and R.~Longo, \emph{Standard and split inclusions of von {N}eumann
			algebras}, Invent. Math. \textbf{75} (1984), 493--536. \MR{735338
			(86h:46095)}
		
		\bibitem[DR89]{DR89}
		S.~Doplicher and J.~E. Roberts, \emph{A new duality theory for compact groups},
		Inventiones Mathematicae \textbf{98} (1989), no.~1, 157--218.
		
		\bibitem[D{\"u}t19]{duetsch2019classical}
		M.~D{\"u}tsch, \emph{From classical field theory to perturbative quantum field
			theory}, Progress in Mathematical Physics, Springer International Publishing,
		2019.
		
		\bibitem[Emc72]{Emch}
		G.G. Emch, \emph{Algebraic methods in statistical mechanics and quantum field
			theory}, Interscience monographs and texts in physics and astronomy,
		Wiley-Interscience, 1972.
		
		\bibitem[Few16]{Fewster_Abh:2016}
		Christopher~J. Fewster, \emph{The split property for quantum field theories in
			flat and curved spacetimes}, Abh. Math. Semin. Univ. Hambg. \textbf{86}
		(2016), no.~2, 153--175. \MR{3561860}
		
		\bibitem[{Few}19]{Fewster:2019}
		Christopher~J. {Fewster}, \emph{{A generally covariant measurement scheme for
				quantum field theory in curved spacetimes}}, arXiv e-prints (2019),
		arXiv:1904.06944.
		
		\bibitem[FR16]{JMPReview}
		K.~Fredenhagen and K.~Rejzner, \emph{Quantum field theory on curved spacetimes:
			{A}xiomatic framework and examples}, Journal of Mathematical Physics
		\textbf{57} (2016), no.~3.
		
		\bibitem[Fre85]{Fredenhagen:1985}
		Klaus Fredenhagen, \emph{On the modular structure of local algebras of
			observables}, Comm. Math. Phys. \textbf{97} (1985), 79--89. \MR{782959
			(86f:81069)}
		
		\bibitem[Fre95]{KlausSuperselection}
		K.~Fredenhagen, \emph{Superselection sectors}, 1995,
		\href{https://unith.desy.de/sites/sites_custom/site_unith/content/e28509/e45341/e47501/e52578/Superselection.pdf}{lecture
			notes}.
		
		\bibitem[FRS89]{FRS89}
		K.~Fredenhagen, K.-H. Rehren, and B.~Schroer, \emph{Superselection sectors with
			braid group statistics and exchange algebras}, Communications in Mathematical
		Physics \textbf{125} (1989), no.~2, 201--226.
		
		\bibitem[FV13]{FewsterVerch:2013}
		Christopher~J. Fewster and Rainer Verch, \emph{The necessity of the {H}adamard
			condition}, Classical Quantum Gravity \textbf{30} (2013), no.~23, 235027, 20.
		\MR{3129302}
		
		\bibitem[FV15]{FewVerchReview}
		C.~J. Fewster and R.~Verch, \emph{Algebraic quantum field theory in curved
			spacetimes}, pp.~125--189, Springer, 2015.
		
		\bibitem[FV18]{FV:2018}
		C.~J. {Fewster} and R.~{Verch}, \emph{{Quantum fields and local measurements}},
		ArXiv e-prints (2018), arXiv:1810.06512.
		
		\bibitem[Haa87]{Haagerup:1987}
		Uffe Haagerup, \emph{Connes' bicentralizer problem and uniqueness of the
			injective factor of type {${\rm III}_1$}}, Acta Math. \textbf{158} (1987),
		95--148. \MR{880070 (88f:46117)}
		
		\bibitem[Haa96]{Haag}
		Rudolf Haag, \emph{Local quantum physics}, second ed., Texts and Monographs in
		Physics, Springer-Verlag, Berlin, 1996, Fields, particles, algebras.
		\MR{1405610}
		
		\bibitem[HM07]{HalvorsenMueger:2006}
		Hans {Halvorson} and Michael {M\"{u}ger}, \emph{{Algebraic Quantum Field
				Theory}}, Philosophy of Physics (J.~Butterfield and J.~Earman, eds.),
		Handbook of the Philosophy of Science, Elsevier Science, 2007,
		arXiv:math-ph/0602036, pp.~731--922.
		
		\bibitem[HS18]{hollands2018entanglement}
		S.~Hollands and K.~Sanders, \emph{Entanglement measures and their properties in
			quantum field theory}, SpringerBriefs in Mathematical Physics, Springer
		International Publishing, 2018.
		
		\bibitem[Kay85]{Kay:1985a}
		Bernard~S. Kay, \emph{A uniqueness result for quasifree {KMS} states}, Helv.
		Phys. Acta \textbf{58} (1985), no.~6, 1017--1029. \MR{821118}
		
		\bibitem[KW91]{KayWald:1991}
		Bernard~S. Kay and Robert~M. Wald, \emph{Theorems on the uniqueness and thermal
			properties of stationary, nonsingular, quasifree states on spacetimes with a
			bifurcate {K}illing horizon}, Phys. Rep. \textbf{207} (1991), no.~2, 49--136.
		\MR{1133130}
		
		\bibitem[Lec15]{Lech_chap:2015}
		Gandalf Lechner, \emph{Algebraic constructive quantum field theory: Integrable
			models and deformation techniques}, Advances in Algebraic Quantum Field
		Theory (Romeo Brunetti, Claudio Dappiaggi, Klaus Fredenhagen, and Jakob
		Yngvason, eds.), Mathematical Physics Studies, Springer International
		Publishing, Springer International Publishing, 2015, pp.~397--448.
		
		\bibitem[Lon89]{Lon89}
		R.~Longo, \emph{Index of subfactors and statistics of quantum fields. {I}},
		Communications in Mathematical Physics \textbf{126} (1989), no.~2, 217--247.
		
		\bibitem[Pei52]{Peierls:1952}
		R.~E. Peierls, \emph{The commutation laws of relativistic field theory}, Proc.
		Roy. Soc. London. Ser. A. \textbf{214} (1952), 143--157. \MR{0051725}
		
		\bibitem[Rej16]{Rejzner_book}
		Kasia Rejzner, \emph{Perturbative algebraic quantum field theory: An
			introduction for mathematicians}, Mathematical Physics Studies, Springer,
		Cham, 2016. \MR{3469848}
		
		\bibitem[RS75]{ReedSimon:vol2}
		Michael Reed and Barry Simon, \emph{Methods of modern mathematical physics.
			{II}. {F}ourier analysis, self-adjointness}, Academic Press [Harcourt Brace
		Jovanovich, Publishers], New York-London, 1975. \MR{0493420}
		
		\bibitem[RS80]{ReedSimon:vol1}
		\bysame, \emph{Methods of modern mathematical physics. {I}}, second ed.,
		Academic Press, Inc. [Harcourt Brace Jovanovich, Publishers], New York, 1980,
		Functional analysis. \MR{751959}
		
		\bibitem[Sim98]{Simon:1998}
		Barry Simon, \emph{The classical moment problem as a self-adjoint finite
			difference operator}, Adv. Math. \textbf{137} (1998), 82--203.
		
		\bibitem[SW00]{StreaterWightman}
		R.~F. Streater and A.~S. Wightman, \emph{P{CT}, spin and statistics, and all
			that}, Princeton Landmarks in Physics, Princeton University Press, Princeton,
		NJ, 2000, Corrected third printing of the 1978 edition. \MR{1884336}
		
		\bibitem[Tho14]{Thomas:2014}
		Erik G.~F. Thomas, \emph{A polarization identity for multilinear maps}, Indag.
		Math. (N.S.) \textbf{25} (2014), no.~3, 468--474, With an appendix by Tom H.
		Koornwinder. \MR{3188841}
		
		\bibitem[VH52]{vanHove:1952}
		L\'{e}on Van~Hove, \emph{Les difficult\'{e}s de divergences pour un mod\`ele
			particulier de champ quantifi\'{e}}, Physica \textbf{18} (1952), 145--159.
		\MR{0049093}
		
		\bibitem[Vla66]{Vladimirov}
		Vasili\u{\i}~Sergeevi\v{c} Vladimirov, \emph{Methods of the theory of functions
			of many complex variables}, Translated from the Russian by Scripta Technica,
		Inc. Translation edited by Leon Ehrenpreis, The M.I.T. Press, Cambridge,
		Mass.-London, 1966. \MR{0201669}
		
		\bibitem[Yng05]{Yngvason:2005}
		Jakob Yngvason, \emph{The role of type {III} factors in quantum field theory},
		Rep. Math. Phys. \textbf{55} (2005), no.~1, 135--147. \MR{2126421}
		
		\bibitem[Yng15]{Yngvason:2015}
		\bysame, \emph{Localization and entanglement in relativistic quantum physics},
		The Message of Quantum Science: Attempts Towards a Synthesis (Philippe
		Blanchard and J{\"u}rg Fr{\"o}hlich, eds.), Springer Berlin Heidelberg,
		Berlin, Heidelberg, 2015, pp.~325--348.
		
		\bibitem[Yuk35]{Yukawa:1935}
		Hideki Yukawa, \emph{{On the Interaction of Elementary Particles I}}, Proc.
		Phys. Math. Soc. Jap. \textbf{17} (1935), 48--57, [Prog. Theor. Phys.
		Suppl.1,1(1935)].
		
\end{thebibliography}
\end{document}